\documentclass[pra,amsmath,amssymb,amsfonts,superscriptaddress]{revtex4-2}
\usepackage{graphicx} % Required for inserting images
\usepackage{appendix}
\usepackage{amsfonts,amssymb}
\usepackage{physics}
\usepackage{tikz}
\usepackage{float}
\usepackage{hyperref}
\usepackage{dsfont}
\usepackage{quantikz}
\usepackage{amsthm}
\usepackage{subfig}
\usepackage{soul}

\newtheorem{theorem}{Theorem}

\newtheorem{corollary}[theorem]{Corollary}

\theoremstyle{definition}

\begin{document}

\title{Quantum Barankin bounds beyond local unbiasedness: Analytic results for Gaussian states via a right-division framework}

\author{Maximilian Reichert}
\affiliation{Department of Physical Chemistry, University of the Basque Country UPV/EHU, Apartado 644, 48080 Bilbao, Spain}
\affiliation{EHU Quantum Center, University of the Basque Country UPV/EHU, Apartado 644, 48080 Bilbao, Spain} 
\author{Javier Navarro}
\affiliation{Department of Physical Chemistry, University of the Basque Country UPV/EHU, Apartado 644, 48080 Bilbao, Spain}
\author{Gerardo Adesso}
\affiliation{School of Mathematical Sciences and Centre for the Mathematics and Theoretical Physics of Quantum Non-Equilibrium Systems, University of Nottingham, University Park, Nottingham NG7 2RD, UK}
\author{Manuel Gessner}
\affiliation{Instituto de Física Corpuscular (IFIC), CSIC-Universitat de València and Departament de Física Teòrica, UV, Avinguda Vicent Andrés Estellés, 19, E-46100 Burjassot (València), Spain}
\author{Mikel Sanz}
\affiliation{Department of Physical Chemistry, University of the Basque Country UPV/EHU, Apartado 644, 48080 Bilbao, Spain}
\affiliation{EHU Quantum Center, University of the Basque Country UPV/EHU, Apartado 644, 48080 Bilbao, Spain} 
\affiliation{Basque Center for Applied Mathematics (BCAM), Alameda de Mazarredo 14, 48009 Bilbao, Spain}

\date{\today}

\begin{abstract}
We propose a quantum version of the Barankin bound as an alternative to the quantum Cram\'er-Rao bound for quantum parameter estimation. The quantum Barankin bound provides a lower bound on the mean squared error of estimators satisfying arbitrarily chosen bias constraints at arbitrarily chosen parameter points. In particular, unbiasedness over the entire parameter space can be imposed, yielding precision limits for globally unbiased quantum parameter estimation. In contrast to the recently derived quantum Barankin bound, which is based on a symmetric division superoperator, our bound is based on a nonsymmetric right-division superoperator. This formulation enables us to find analytic expressions for the Barankin matrix for Gaussian states, which allows for an efficient calculation of the bound. We demonstrate the usefulness of our results by applying them to various examples that exhibit the threshold effect in the few-shot regime, which is invisible to the standard quantum Cram\'er-Rao bound approach.
\end{abstract}

\maketitle

\section{Introduction}
Quantum parameter estimation theory is concerned with determining the ultimate precision limits of quantum estimation protocols \cite{helstrom1969a,paris2009a,giovannetti2011,toth2014,liu2019}. A key figure of merit in the field is the quantum Cram\'er-Rao bound (QCRB), which is a lower bound on  the  mean squared error (MSE) for any \emph{locally} unbiased estimator. The precision set by the QCRB can be attained with sufficiently good prior knowledge of the parameter or in the regime of many repetitions of the experiment. In realistic scenarios, often neither of the aforementioned two conditions are satisfied, and thus, the actual estimation performance can be far worse than what the QCRB suggests. In these situations, it is desirable to have tighter, more informative bounds. One route is to employ Bayesian methods~\cite{rubio2019,rubio2020a,rubio2018,li2018,albarelli2026}, which lead to the derivations of various Bayesian bounds including  Weiss-Weinstein \cite{weiss1985,weinstein1988,lu2016} and Ziv-Zakai bounds \cite{ziv1969,tsang2012}. Another approach is the frequentist framework, which we consider in this paper. In the classical  frequentist framework, tighter generalizations of the CRB have been  proposed. Two prominent examples are the families of Barankin~\cite{barankin1949}  and Bhattacharyya~\cite{bhattacharyya1946} bounds. The Bhattacharyya bounds generalize the CRB by extending the estimator's bias constraints from the first derivative to derivatives of arbitrary order. The Barankin bounds, on the other hand, impose bias constraints  at a freely chosen set of parameter test points. A particularly natural choice for the bias constraints is to impose unbiasedness at the test points. With this, by optimizing over the number and position of test points, one obtains the greatest lower bound on the MSE, evaluated at the true parameter value, among all estimators that are unbiased at all possible parameter values. Barankin showed that this lower bound is attainable if it exists \cite{barankin1949}.

 Recently, Ref.~\cite{gessner2023hierarchies} derived quantum versions of the Barankin  and Bhattacharyya bounds. The derivation relies on symmetric division superoperators, of which the symmetric logarithmic derivative (SLD) relevant for the QCRB is a special case. While explicit formulas to calculate these general quantum bounds were provided, the evaluation seems challenging for high or infinite dimensional states, such as Gaussian states, as the formulas rely on finding the eigendecomposition of the quantum state. For Gaussian states, no analytic formulas in terms of mean vector and covariance matrix are known.

 In this paper, we derive a novel quantum Barankin bound, which is based on a non-symmetric ``right'' division superoperator. Our derivation is general in the sense that it also applies to infinite dimensional quantum systems. We provide explicit sufficient existence conditions for the infinite dimensional case. Furthermore, we establish a hierarchy between our non-symmetric quantum Barankin bound and the symmetric version and show that the non-symmetric bound is always smaller than the symmetric version, in analogy to the QCRB hierarchy between the right logarithmic derivative (RLD) and the SLD version. We furthermore show that our quantum Barankin bound recovers the QCRB based on the RLD under suitable regularity conditions in the limit of only two test points of vanishing distance.
 
 A key advantage of our approach is that for Gaussian states, we found analytic formulas for the Barankin matrix in terms of the mean vectors and covariance matrices in analogy to Gaussian state formulas for the QCRB~\cite{safranek2018a,safranek2015,gao2014bounds,bakmou2020multiparameter,monras2013,jiang2014,sorelli2024a,shoukang2026}. This allows for an efficient evaluation of the bound even for highly multi-mode Gaussian states.

\section{Preliminaries}

Before starting, let us note that vectors are printed bold, quantum operators carry a hat, the matrix transpose is denoted as $()^T$, the Hermitian conjugate is denoted as $()^\dag$ and the complex conjugate is denoted as $()^*$.

\subsection{Classical Estimation Theory}\label{subsec:ClassicalEstimation}
We begin with a brief review of a general framework of local classical estimation theory that contains as special cases the well-known Cram\'er-Rao bound and the Barankin bound. Our derivation follows Refs.~\cite{gessner2023hierarchies,Reichert2026}.

In local estimation theory the goal is to estimate a fixed parameter $\lambda \in \Lambda$ (the true parameter) from a set $\Lambda \subset \mathds{R}$ of possible parameter values  by observing a sample $\mathbf x\in \mathds{R}^d$ (a measurement outcome in the context of physics) drawn from the probability density function $p(\mathbf x\vert\lambda)$. For this, an estimator $\lambda_{\text{est}}(\mathbf x)$  maps the observed data $\mathbf x$ to an estimate. As a performance metric for estimators we use the mean squared error (MSE), which is defined as
\begin{align}
  \text{MSE}_\lambda[\lambda_{\text{est}}]   =  \mathds{E}_\lambda\left[ \left(\lambda_{\text{est}}(\mathbf X) - \lambda \right)^2\right] = \mathrm{Var}_\lambda[\lambda_{\text{est}}(\mathbf{X})] + b_\lambda^2 ,
\end{align}
where $\mathds{E}_{\lambda}[(\cdot)] = \int_D \mathrm{d}^d x \, p(\mathbf x\vert\lambda)  (\cdot) $. We integrate over the support $D = \overline{\{ \mathbf x : p(\mathbf x\vert \lambda) > 0\}}$ of the probability density $p(\mathbf x\vert\lambda)$. In the discrete case, the integral is understood as a sum over outcomes $\mathds{E}_{\lambda}[(\cdot)] = \sum_{\mathbf x\in D} p(\mathbf x\vert \lambda) (\cdot)$.

We also note that $ \mathrm{Var}_\lambda[\lambda_{\text{est}}(\mathbf{X})] = \mathds{E}_{\lambda}[(\lambda_{\text{est}}(\mathbf{X})-\mathds{E}_{\lambda}[\lambda_{\text{est}}(\mathbf{X})])^2]$ is the variance and
\begin{align}
    b_{\lambda} = \mathds{E}_{\lambda} [\lambda_{\text{est}}(\mathbf{X})] -\lambda
\end{align}
is the bias of the estimator at the true value of the parameter $\lambda$.

We now derive a family of lower bounds on the MSE for estimators that satisfy certain constraints.
For this, let us introduce the complex vector  function $\mathbf g (\mathbf x) = (g_1(\mathbf x),\ldots,g_K(\mathbf x))^T$ , and the complex vector $\mathbf a = (a_1,\ldots, a_K)^T$.
 With this, we define the functions
 $
    v(\mathbf x) =  \sqrt{p  (\mathbf x\vert\lambda)} \left(\lambda_{\text{est}}(\mathbf x) - \lambda \right) 
$
and 
 $
    u(\mathbf x) = \frac{1}{\sqrt{p  (\mathbf x\vert\lambda)}} \mathbf a^\dag \mathbf g(\mathbf x)
 $, defined on the support $D$ of $p(\mathbf x\vert \lambda)$.
Let  $u(\mathbf x),v(\mathbf x) \in L^2 (D)$, then 
the Cauchy-Schwarz inequality holds
   $ \left \vert \int \mathrm{d}^d x \, u^*(\mathbf x) v(\mathbf x)\right\vert^2 \leq \left( \int \mathrm{d}^d x \, \vert u(\mathbf x) \vert^2\right)\left( \int \mathrm{d}^d x \, \vert v(\mathbf x) \vert^2\right)$. With this, we obtain the general lower bound on the MSE
\begin{align}
    \text{MSE}_\lambda [\lambda_{\text{est}}] \geq \frac{\vert \mathbf a^\dag\boldsymbol\gamma \vert^2}{\mathbf a^\dag C\mathbf a} ,  
\end{align}
for estimators $\lambda_{\text{est}} (\mathbf x)$ that satisfy the constraint equations
\begin{align}
    \int_{D} \mathrm{d}^d x \,  \mathbf g(\mathbf x) ( \lambda_{\text{est}} (\mathbf x)- \lambda) = \boldsymbol{\gamma}\label{sec3:constraintgamma} ,
\end{align}
where $\boldsymbol\gamma \in\mathds{C}^K$.
Here, $C$ is a Hermitian positive semidefinite $K\times K$ matrix with elements
\begin{align}
    C_{kl} = \int_{D} \mathrm{d}^d x \, \frac{g_k(\mathbf x)g_l^*(\mathbf x)}{p(\mathbf x\vert\lambda)} \label{eq:InformationMatrix} .
\end{align}

Now, for given $\mathbf{g}(\mathbf x)$ and $\boldsymbol \gamma$ we can optimize over $\mathbf a\in \mathds{C}^K$ which yields
\begin{align}
   \text{MSE}_\lambda [\lambda_{\text{est}}] \geq \sup_{\mathbf a\neq 0} \frac{\vert \mathbf a^\dag\boldsymbol\gamma \vert^2}{\mathbf a^\dag C\mathbf a} \label{eq:GeneralBound} .
\end{align}

Assuming that $C$ is invertible, one can again use the Cauchy-Schwarz inequality for vectors $\mathbf w,\mathbf q\in \mathds{C}^K$: $\vert \mathbf q^\dag \mathbf w\vert^2 \leq \Vert \mathbf q\Vert^2 \Vert \mathbf w\Vert^2$. Now, set $\mathbf q = C^{1/2} \mathbf a$ and $\mathbf w = C^{-1/2}\boldsymbol \gamma$. With this, we find
$
    \frac{\vert \mathbf a^\dag\boldsymbol\gamma \vert^2}{\mathbf a^\dag C\mathbf a} =  \frac{\vert \mathbf q^\dag\mathbf w \vert^2}{\mathbf a^\dag C\mathbf a} \leq \boldsymbol\gamma^\dag C^{-1} \boldsymbol \gamma .
$
We have equality above, if and only if $\mathbf w$ and $\mathbf q$ are parallel, which is the case for the choice $\mathbf a = C^{-1} \boldsymbol \gamma$. Thus, we arrive at the optimized bound
\begin{align}
    \text{MSE}_\lambda [\lambda_{\text{est}}] \geq \boldsymbol\gamma^\dag C^{-1} \boldsymbol \gamma \label{eq:GeneralBoundInvertibile}.
\end{align}
In the case in which the inverse does not exist, we can replace the $C^{-1}$ above by the Moore-Penrose pseudo inverse if $\boldsymbol \gamma $ is in the range of $C$. If $C$ is not invertible and $\boldsymbol\gamma$ is not in the range of $C$, the supremum tends to infinity and the bound becomes meaningless. Analogous discussions regarding invertibility apply to the quantum versions of the quantum Barankin bounds that we discuss later.

\subsubsection{Cram\'er-Rao bound}
Let us now recover the well-known CRB from this very general approach.
 For this, one chooses the two test functions  $\mathbf g(\mathbf x) = ( g_1 (\mathbf x), g_2(\mathbf x))^T  $ with $g_1(\mathbf x) = p  (\mathbf x\vert\lambda)$ and $g_2(\mathbf x) = \partial_{\tilde\lambda} p  (\mathbf x\vert \tilde\lambda) \vert_{\tilde\lambda = \lambda}$. We assume that the usual regularity conditions are satisfied. For this choice of $\mathbf g(\mathbf x)$, we can calculate the matrix in Eq.~\eqref{eq:InformationMatrix} and obtain
\begin{align}
    C_{\text{CRB}} = \begin{pmatrix}
        1 & 0 \\
        0 & F (\lambda )  
    \end{pmatrix} ,
\end{align}
where
\begin{align} \label{sec3:FIdefinition}
    F(\lambda) &= \int_{D}\mathrm{d}^dx\, \frac{(\partial_{\tilde\lambda} p  (\mathbf x\vert\tilde\lambda) \vert_{\tilde\lambda = \lambda})^2}{p (\mathbf x\vert\lambda)} 
\end{align}
is the Fisher information.
With this follows a general form of the CRB
\begin{align}
    \text{MSE}_{\lambda}[\lambda_{\text{est}}]  \geq \boldsymbol \gamma^\dag C^{-1} \boldsymbol \gamma &= \gamma_1^2 \cdot 1 + \frac{\gamma_2^2}{F(\lambda)}  \\
    &= b_\lambda^2 + \frac{(1+\partial_{\tilde\lambda} b_{\tilde\lambda} \vert_{\tilde\lambda = \lambda})^2}{F(\lambda)} ,
\end{align}
with $\gamma_1^2 = b_\lambda^2$ and $\gamma_2^2= (1+\partial_{\tilde\lambda} b_{\tilde\lambda} \vert_{\tilde\lambda = \lambda})^2$.
Note that the bound explicitly depends on the bias and its derivative at the true parameter $\lambda$. When setting $b_\lambda=\partial_{\tilde\lambda} b_{\tilde\lambda} \vert_{\tilde\lambda = \lambda} = 0$, we recover the most known form
\begin{align}
      \text{MSE}_\lambda[\lambda_{\text{est}}]  \geq  \frac{1}{F(\lambda)}
\end{align}
of the CRB valid for all locally unbiased estimators.

\subsubsection{Barankin bound}
Let us now recover the classical Barankin bound from our general framework. The Barankin bound imposes bias constraints at test points $\{\lambda_k\}_{k=1}^K$ freely chosen from the set  $\Lambda$. This set $\Lambda$ usually also contains $\lambda$.
As the test functions $\mathbf g (\mathbf x) = (g_1(\mathbf x),\ldots, g_K(\mathbf x))^T$ one chooses $g_k(\mathbf x) = p(\mathbf{ x}\vert\lambda_k)$, i.e., the probability density evaluated at the parameter test points. For this choice, we obtain the Barankin bound
\begin{align}
    \text{MSE}_\lambda [\lambda_{\text{est}}] \geq \frac{\vert \mathbf a^\dag\boldsymbol\gamma \vert^2}{\mathbf a^\dag C_{\mathrm{BB}}\mathbf a} ,  
\end{align}
and further we have $ \text{MSE}_\lambda [\lambda_{\text{est}}]\geq \mathbf{\gamma}^\dag C_{\mathrm{BB}}^{-1}\mathbf{\gamma}$ if $C_{\mathrm{BB}}$ invertible. For the case of $C_{\mathrm{BB}}$ not invertible, refer to the discussion in Section~\ref{subsec:ClassicalEstimation}

Note that $C_{\mathrm{BB}}$ is the Barankin matrix with elements
\begin{align}
   (C_{\mathrm{BB}})_{kl}= \int_{D}\mathrm{d}^dx\,\frac{p(\mathbf x\vert\lambda_k)p(\mathbf x\vert\lambda_l)}{p(\mathbf x\vert\lambda)} ,
\end{align}
and the estimators to which the Barankin bound applies have   to satisfy the bias constraints
\begin{align}
    \gamma_k = \int_{D}\mathrm{d}^dx\, p  (\mathbf x\vert\lambda_k) (\lambda_{\text{est}}(\mathbf x) - \lambda) = \mathds{E}_{\lambda_k} [\lambda_{\text{est}}] - \lambda = b_{\lambda_k} +\lambda_k-\lambda ,
\end{align}
where $b_{\lambda_k}$ is the bias of the estimator $\lambda_{\text{est}}$ if the value of the parameter were $\lambda_k$. 

A natural choice for the constraints is $\gamma_k = \lambda_k-\lambda$, which forces unbiasedness $b_{\lambda_k}=0$ at all test points. For a given finite set of test points, one obtains a lower bound on the MSE (evaluated at $\lambda$) for all estimators unbiased at these test points.
By taking the supremum of this bound over all finite collections of test points, i.e., over both  their locations and their number, one obtains a lower bound on the MSE, evaluated at the true value of $\lambda$, that applies to all estimators that are globally unbiased on all of $\Lambda$. Barankin showed in Ref.~\cite{barankin1949} that, under the appropriate existence conditions, this supremum characterizes the MSE of the best unbiased estimator on $\Lambda$ evaluated at $\lambda$. This optimization is often hard to perform analytically. In practice, one can use a sequence of nested and increasingly dense test points and compute the corresponding bounds. Under suitable regularity conditions, one will observe convergence to the optimized bound.
Due to these convergence properties, we will only study the Barankin bounds over discrete sets of test points, although a generalization to a possibly uncountable set of test points can be done straightforwardly.

\subsection{Quantum estimation theory}
Thus far we have only considered classical estimation theory. This theory can readily be extended to the quantum case, and we will give a brief review on existing results.

First, we consider a family of quantum states $\{\hat\rho_{\tilde \lambda}\}_{\tilde\lambda\in\Lambda}$, where $\tilde\lambda$ labels the state. The actual state of the system is $\hat\rho_\lambda$, corresponding to an unknown but fixed value $\lambda\in\Lambda$. To make contact to the classical theory, let us introduce positive-operator valued measure (POVM) elements $\{\hat \Pi_{\mathbf x}\}_{\mathbf x}$ with $\int_D\mathrm{d}^d x\, \hat \Pi_{\mathbf x}  = \hat I$ which represent the measurement, so that
\begin{align}
    p(\mathbf x\vert \lambda) = \mathrm{Tr}[\hat\rho_\lambda \hat\Pi_{\mathbf{x}}]
  \end{align}
is a classical probability density. Maximizing the classical FI $F(\lambda)$ over $\{\hat \Pi_{\mathbf x}\}_{\mathbf x}$ yields the quantum Fisher information \cite{paris2009a}
\begin{align}
      \mathcal{F}^{(S)}(\hat\rho_\lambda)   =\max_{\{\hat \Pi_{\mathbf x}\}_{\mathbf x}}F(\lambda)=\mathrm{Tr}[\hat \rho_\lambda (\hat L_\lambda ^{(S)})^2],
\end{align}
where $\hat L_\lambda ^{(S)} = \hat\Omega_{\hat \rho_\lambda} (\partial_{\tilde\lambda}\hat\rho_{\tilde\lambda}\vert_{\tilde\lambda=\lambda})$ is the symmetric logarithmic derivative (SLD) based on the symmetric division operator $\hat \Omega_{\hat A}(\hat B)$ implicitly defined via the equation
\begin{align}
    \hat B= \frac{\hat A \hat \Omega_{\hat A}(\hat B)+\hat \Omega_{\hat A}(\hat B) \hat A}{2} .\label{eq:SymmetricDivision}
\end{align}
The QFI based on the SLD is not the only quantum generalization of the FI. One can also define a QFI based on the right logarithmic derivative (RLD) \cite{yuen1973multiple}
\begin{align}
       \mathcal{F}^{(R)}(\hat\rho_\lambda) = \mathrm{Tr}[\hat \rho_\lambda  \hat L_\lambda ^{(R)}  (\hat L_\lambda ^{(R)})^\dag ] ,
\end{align}
where $\hat L_\lambda ^{(R)} = \hat R_{\hat \rho_{\lambda}}(\partial_{\tilde\lambda}\hat\rho_{\tilde\lambda}\vert_{\tilde\lambda=\lambda})$ and the right division operator is defined via
\begin{align}
    \hat B = \hat A\hat R_{\hat A}(\hat B)  . \label{eq:RightDivision}
\end{align}
If $\hat A$ is faithful, the superoperator $\hat R_{\hat A}$ exists (on a suitable domain) and is unique and linear. For our upcoming derivations in Section~\ref{sec:GeneralResult}, the operator $\hat R_{\hat A}$ never appears isolated, and always appears in the form $\hat A^{1/2}\hat R_{\hat A}(\hat B)$. Again, this operator is only defined on a certain domain, but in the following derivations in Section~\ref{sec:GeneralResult}, we will assume that $\hat A^{1/2}\hat R_{\hat A}(\hat B)$ can be extended to a trace class operator on the full domain and we will use the same symbol $\hat A^{1/2}\hat R_{\hat A}(\hat B)$ for this extension.

The QFIs based on the SLD and RLD belong to a larger family of QFIs based on a family of Riemannian monotone metrics \cite{petz1996}.
These different QFIs lead to different quantum versions of the Cram\'er-Rao bound, however, with different attainability properties. In the multi-parameter case, which we will not study here, there is generally no clear hierarchy between the SLD and RLD based bounds, but it can be shown that for the special class of D-invariant models the RLD QCRB is the only one attainable \cite{suzuki2016,liu2019}.
In the single-parameter case, there is a clear hierarchy, and we have
\begin{align}
     \text{MSE}_\lambda[\lambda_{\text{est}}]   \geq \frac{1}{ \mathcal{F}^{(S)}(\hat\rho_\lambda)} \geq \frac{1}{ \mathcal{F}^{(R)}(\hat\rho_\lambda)} .\label{eq:CRBhierarchy}
\end{align}
The single-parameter SLD QCRB can always be attained by measuring the observable $\lambda_{\text{pr}} \hat I + \hat L_{\lambda_{\text{pr}}}^{(S)}/\mathcal{F}^{(S)}(\hat\rho_{\lambda_\text{pr}})$, provided the prior guess $\lambda_{\text{pr}}$ of $\lambda$ is good enough. From this, it follows that the RLD QCRB can generally not be attained and is less informative than the SLD version in the single-parameter case.

Recently, Ref.~\cite{gessner2023hierarchies} has derived a quantum version of  the general classical bound given in Eq.~\eqref{eq:GeneralBound} that we discussed in Section~\ref{subsec:ClassicalEstimation}. For this, they used the symmetric division operator from Eq.~\eqref{eq:SymmetricDivision}. Particularly, they derived a quantum Barankin bound $\mathrm{MSE}_\lambda[\lambda_{\mathrm{est}}] \geq \sup_{\mathbf a\neq 0} \frac{\vert \mathbf a^\dag\boldsymbol\gamma \vert^2}{\mathbf a^\dag Q_{\text{BB}}^{(S)}\mathbf a} $ based on the symmetric division operator with quantum Barankin matrix
\begin{align}
    (Q_{\text{BB}}^{(S)})_{kl} = \mathrm{Tr}[\hat\rho_{\lambda_k} \hat\Omega_{\hat\rho_\lambda}(\hat\rho_{\lambda_l})] .
\end{align}

A natural question to ask is whether one can derive a quantum Barankin bound based on the ``right'' division operator
\begin{align}
    \hat R_k := \hat R_{\hat\rho_\lambda}(\hat\rho_k) . \label{eq:AbbreviationRightDivision}
\end{align}
 In this paper, we answer this question in the affirmative.

\section{Derivation of Quantum Barankin bound based on right division operator} \label{sec:GeneralResult}
Let us now derive a quantum Barankin bound based on the right division operator, which is a different  quantum version of the quantum Barankin bound based on the symmetric division operator derived in Ref.~\cite{gessner2023hierarchies}. We compare these two bounds in Section~\ref{subsec:Attainability}.

Our focus in this paper is on infinite-dimensional Hilbert spaces, but note that our results also hold for the finite dimensional case.
In Section~\ref{sec:Gaussian}, we consider bosonic Gaussian states and will derive analytic formulas for the Barankin matrix elements in terms of their mean vectors and covariance matrices. Comparable formulas do not exist for the symmetric version of Ref.~\cite{gessner2023hierarchies}.

\subsection{Main result}

\begin{theorem}\label{theorem:1}
    Let  $\{\hat \rho_{\tilde\lambda}\}_{\tilde\lambda\in\Lambda}$ be a family of quantum states  and let $\hat\rho_\lambda$ be the true quantum state, where   $\lambda$ is the (fixed but unknown) value of the parameter to be estimated. Assume $\hat\rho_\lambda$ to be faithful. Let $\{\lambda_k \}_{k=1}^K$ be parameter test points with $\lambda_k\in\Lambda$ so that $ \hat\rho_\lambda^{1/2} \hat R_k$ admits a trace-class extension, where $\hat R_k =\hat R_{\hat\rho_\lambda}(\hat\rho_k)$ as defined in Eq.~\eqref{eq:AbbreviationRightDivision}. 
    
    It then follows that for every pair of POVM $\{\hat \Pi_{\mathbf x}\}_{\mathbf x }$ and estimator $\lambda_{\mathrm{est}}$ that satisfy the bias constraints 
    \begin{align}
    \int_D \mathrm{d}^d x \, \mathrm 
    {Tr}[\hat \rho_{\lambda_k}\hat \Pi_{\mathbf x}] \lambda_{\mathrm{est}}(\mathbf x) = \lambda_k+b_{\lambda_k},
    \end{align}
    we have
    \begin{align}
         \mathrm{MSE}_\lambda[\lambda_{\mathrm{est}}] \geq \sup_{\mathbf{a}\neq 0} \frac{\vert\mathbf{a}^\dag \boldsymbol{\gamma}\vert^2}{\mathbf{a}^\dag Q_{\mathrm{BB}}^{(R)} \mathbf{a}} 
    \end{align}
    with $\mathbf{a}\in \mathds{C}^K$, $\gamma_k = b_{\lambda_k }+\lambda_k-\lambda$  and
    \begin{align}\label{eq:RBaB}
        (Q_{\mathrm{BB}}^{(R)})_{kl} = \mathrm{Tr} \left[ \hat R_k^\dag\hat\rho_{\lambda}  \hat R_l\right] = \mathrm{Tr}[\hat\rho_{\lambda_k}\hat\rho_{\lambda}^{-1}\hat\rho_{\lambda_l}]  .
    \end{align}
     If additionally $Q_{\mathrm{BB}}^{(R)}$ is invertible~\cite{BarankinMatrix}, we have
    \begin{align}
         \mathrm{MSE}_\lambda[\lambda_{\mathrm{est}}] \geq \boldsymbol \gamma^\dag (Q_{\mathrm{BB}}^{(R)})^{-1} \boldsymbol \gamma .
    \end{align}
\end{theorem}
\begin{proof}
To keep notation compact, let us introduce the abbreviations
$  \hat \rho_\lambda =: \hat \rho$, $\hat \rho_{\lambda_k} =: \hat \rho_k $ inside this proof.

As a starting point we use the classical Barankin bound where the test functions are chosen as  $g_k (\mathbf x) = p(\mathbf x\vert \lambda_k)$. To connect this to the quantum description, we write $g_k (\mathbf x) = \text{Tr} [ \hat{\Pi}_{\mathbf x} \hat{\rho}_k]$. Now, we have
\begin{align}
    \mathbf{a}^\dag C_{\mathrm{BB}} \mathbf{a} &= \sum_{k,l=1}^Ka_k^* \int_D \mathrm{d}^d x\, \frac{\text{Tr} [ \hat{\Pi}_{\mathbf x} \hat{\rho}_k] \text{Tr} [ \hat{\Pi}_{\mathbf x} \hat{\rho}_l]}{\text{Tr} [ \hat{\Pi}_{\mathbf x} \hat{\rho}]} a_l = \int_D \mathrm{d}^d x\,\sum_{k,l=1}^Ka_k^*  \frac{\text{Tr} [ \hat{\Pi}_{\mathbf x} \hat{\rho}_k] \text{Tr} [ \hat{\Pi}_{\mathbf x} \hat{\rho}_l]}{\text{Tr} [ \hat{\Pi}_{\mathbf x} \hat{\rho}]} a_l  \\
    &=   \int_D \mathrm{d}^d x\, \frac{\left\vert\sum_{k=1}^K a_k \text{Tr}[ \hat{\Pi}_{\mathbf x} \hat \rho_k] \right\vert^2}{\text{Tr} [ \hat{\Pi}_{\mathbf x} \hat{\rho}    ]}  = \int_D \mathrm{d}^d x\, \frac{\left\vert \text{Tr}[\sum_{k=1}^K a_k  \hat{\Pi}_{\mathbf x} \hat \rho_k] \right\vert^2}{\text{Tr} [ \hat{\Pi}_{\mathbf x} \hat{\rho}   ]} .
\end{align}
In the first line we interchange the sum with the integral, which is justified by  the linearity of the integral and the finiteness of the sum. In the second line we rewrite the integrand and then interchange the sum with the trace, which is justified by the linearity of the trace and the finiteness of the sum.

With this we have
\begin{align}
     \mathbf{a}^\dag C_{\mathrm{BB}} \mathbf{a} &= \int_D \mathrm{d}^d x\, \frac{\left\vert \text{Tr}\left[ \sum_{k=1}^K a_k \hat{\Pi}_{\mathbf x}^{1/2} \hat{\Pi}_{\mathbf x}^{1/2}\hat \rho_k\right] \right\vert^2}{\text{Tr}\left[ \hat{\Pi}_{\mathbf x}  \hat{\rho} \right]} = \int_D \mathrm{d}^d x\, \frac{\left\vert \text{Tr}\left[ \sum_{k=1}^K a_k\hat{\Pi}_{\mathbf x}^{1/2} \hat \rho_k \hat{\Pi}_{\mathbf x}^{1/2}\right] \right\vert^2}{\text{Tr}\left[ \hat{\Pi}_{\mathbf x}  \hat{\rho} \right]} \\
     &= \int_D \mathrm{d}^d x\, \frac{\left\vert \text{Tr}\left[ \sum_{k=1}^K a_k \hat{\Pi}_{\mathbf x}^{1/2} \hat \rho^{1/2}  \hat \rho^{1/2} \hat R_k \hat{\Pi}_{\mathbf x}^{1/2}\right] \right\vert^2}{\text{Tr}\left[ \hat{\Pi}_{\mathbf x}  \hat{\rho} \right]}
     =  \int_D \mathrm{d}^d x\, \frac{\left\vert \text{Tr}\left[  \left(\hat \rho^{1/2}\hat{\Pi}_{\mathbf x}^{1/2}\right)^{\dag}   \left( \sum_{k=1}^K a_k \hat \rho^{1/2} \hat R_k \hat{\Pi}_{\mathbf x}^{1/2}\right)\right] \right\vert^2}{\text{Tr}\left[ \hat{\Pi}_{\mathbf x}  \hat{\rho} \right]}\\
     &\leq \int_D \mathrm{d}^d x\, \frac{\text{Tr}\left[  \left(\hat \rho^{1/2}\hat{\Pi}_{\mathbf x}^{1/2}\right)^{\dag}   \rho^{1/2}\hat{\Pi}_{\mathbf x}^{1/2}\right]}{\text{Tr}\left[ \hat{\Pi}_{\mathbf x}  \hat{\rho} \right]}  \text{Tr}\left[    \left( \sum_{k=1}^K a_k   \hat \rho^{1/2} \hat R_k \hat{\Pi}_{\mathbf x}^{1/2}\right)^{\dag} \left( \sum_{l=1}^K a_l   \hat \rho^{1/2} \hat R_l \hat{\Pi}_{\mathbf x}^{1/2}\right)\right]   \\
     &= \int_D \mathrm{d}^d x\,   1 \cdot \sum_{k,l=1}^K a_k^* a_l \text{Tr}[ \hat\Pi_{\mathbf x}^{1/2} \hat R_k^\dag \hat\rho^{1/2}  \hat\rho^{1/2}\hat R_l \hat\Pi_{\mathbf x}^{1/2}]  =     \sum_{k,l=1}^K a_k^* a_l\int_D \mathrm{d}^d x\,   \text{Tr}[ \hat\Pi_{\mathbf x} \hat R_k^\dag \hat\rho^{1/2}  \hat\rho^{1/2}\hat R_l ] \label{eq:33}\\
    & =    \sum_{k,l=1}^K a_k^* a_l \text{Tr}[\int_D \mathrm{d}^d x\, \hat\Pi_{\mathbf x} \hat R_k^\dag    \hat\rho \hat R_l ] \\
    & =     \sum_{k,l=1}^K a_k^* a_l \text{Tr}[  \hat R_k^\dag    \hat\rho \hat R_l] \\
    &=  \mathbf a^\dag Q_{\mathrm{BB}}^{(R)} \mathbf a.
\end{align} 
In the first line, we introduce the unique positive semidefinite square root $\hat \Pi_{\mathbf x}^{1/2}$ of the POVM element $\hat \Pi_{\mathbf x}$, both of which are bounded operators. Because $\hat \rho_k$ is trace class and $\hat \Pi_{\mathbf x}^{1/2}$ is bounded, $\hat \Pi_{\mathbf x}^{1/2} \hat \rho_k $ is trace class and we can use the cyclic property of the trace with the remaining $\hat \Pi_{\mathbf x}^{1/2}$. Note that $\mathrm{Tr}[\hat A \hat B] = \mathrm{Tr}[\hat B \hat A ]$ if $\hat A$ trace class and $\hat B$ bounded. 
In the second line we introduce the right division operator $\hat\rho_k = \hat \rho\hat R_k = \hat \rho^{1/2}\hat \rho^{1/2}\hat R_k $. Note that $\hat R_k$ exists because of the assumption that $\hat \rho$ is faithful. Also, by the sufficient assumption, the operator $\hat{\rho}^{1/2}\hat{R}_k$ is trace class, or more accurately, admits a trace-class extension to the entire Hilbert space, which we denote again by the same symbol $\hat{\rho}^{1/2}\hat{R}_k$. Then we regroup  the terms inside the trace using $(\hat \rho^{1/2})^\dag = \hat \rho^{1/2}$ and $(\hat \Pi_{\mathbf x}^{1/2})^\dag = \hat \Pi_{\mathbf x}^{1/2}$. 
In the third line, we use $\vert \mathrm{Tr}[\hat A^\dag \hat B]\vert^2 \leq \mathrm{Tr}[\hat A^\dag \hat A]\mathrm{Tr}[\hat B^\dag \hat B]$, which is valid when $\hat A,\hat B$ are Hilbert-Schmidt operators, i.e., $\mathrm{Tr}[\hat A^\dag \hat A]< \infty$ and $\mathrm{Tr}[\hat B^\dag \hat B]< \infty$. In our case, we have $\hat A= \left(\hat \rho^{1/2}\hat{\Pi}_{\mathbf x}^{1/2}\right)^{\dag}$ which is evidently Hilbert-Schmidt as $\hat A^\dag \hat A = \hat{\Pi}_{\mathbf x}^{1/2} \hat \rho \hat{\Pi}_{\mathbf x}^{1/2}$ is trace class. We further have $\hat B =   \sum_{k=1}^K a_k \hat \rho^{1/2} \hat R_k \hat{\Pi}_{\mathbf x}^{1/2}$. By assumption, $\hat \rho^{1/2} \hat R_k$ is trace class, and thus also $\hat \rho^{1/2} \hat R_k \hat \Pi_{\mathbf x}^{1/2}$ as $\hat \Pi_{\mathbf x}^{1/2}$ is bounded. A finite linear combination of trace class operators is trace class, and thus $\hat B$ is trace class and therefore also Hilbert-Schmidt. 
In the fourth line, we used $\text{Tr}\left[  \left(\hat \rho^{1/2}\hat{\Pi}_{\mathbf x}^{1/2}\right)^{\dag}   \rho^{1/2}\hat{\Pi}_{\mathbf x}^{1/2}\right] = \text{Tr}\left[      \hat{\Pi}_{\mathbf x}^{1/2} \hat \rho^{1/2}\rho^{1/2}\hat{\Pi}_{\mathbf x}^{1/2}\right] =\text{Tr}\left[      \hat{\Pi}_{\mathbf x}  \hat \rho   \right]$, using the self adjointness of $\hat \rho^{1/2}$ and $\hat{\Pi}_{\mathbf x}^{1/2}$, and the fact that $\hat{\Pi}_{\mathbf x}^{1/2} \hat \rho^{1/2}\rho^{1/2}$ is trace class and $\hat{\Pi}_{\mathbf x}^{1/2}$ is bounded which allows us to use the cyclic property of the trace. In the second term, we used the self adjointness of $\hat \rho^{1/2}$ and $\hat\Pi_{\mathbf x}^{1/2}$ and interchanged the sums with the trace, justified by the finiteness of the sum and the fact that $ \hat\Pi_{\mathbf x}^{1/2} \hat R_k^\dag \hat\rho^{ 1/2}  \hat\rho^{ 1/2}\hat R_l \hat\Pi_{\mathbf x}^{1/2}$ is trace class, which follows by assumption that $\hat \rho^{1/2}\hat R_r$ is trace class from which also follows that $ \hat R_r^\dag\hat \rho^{1/2}$ is trace class. After the second equals sign in the fourth line, we interchange the integration and sum, and we also use the cyclic property   of the trace as $\hat\Pi_{\mathbf x}^{1/2} \hat R_k^\dag \hat\rho^{1/2}  \hat\rho^{1/2}\hat R_l $ is trace class and $\hat\Pi_{\mathbf x}^{1/2}$ bounded.
In the fifth line we pulled the integral into the trace, which is valid as the operator inside the trace is trace class. 
In the sixth line we used the completeness of the POVM elements.

We thus have shown that
\begin{align}
   \mathbf a^\dag C_{\mathrm{BB}}  \mathbf a \leq \mathbf a^\dag Q_{\mathrm{BB}}^{(R)} \mathbf a
\end{align}
from which follows  
\begin{align}
     \mathrm{MSE}[\lambda_{\mathrm{est}}] \geq \sup_{\mathbf{a}\neq 0} \frac{\vert\mathbf{a}^\dag \boldsymbol{\gamma}\vert^2}{\mathbf{a}^\dag C_{\mathrm{BB}}  \mathbf{a}}  \geq \sup_{\mathbf{a}\neq 0} \frac{\vert\mathbf{a}^\dag \boldsymbol{\gamma}\vert^2}{\mathbf{a}^\dag Q_{\mathrm{BB}}^{(R)} \mathbf{a}} .
\end{align} 
Assuming that $Q_{\mathrm{BB}}^{(R)}$ is invertible, we find following the same reasoning as in Section~\ref{subsec:ClassicalEstimation} that
\begin{align}
      \mathrm{MSE}[\lambda_{\mathrm{est}}] \geq  \boldsymbol\gamma^\dag (Q_{\mathrm{BB}}^{(R)})^{-1} \boldsymbol\gamma .
\end{align}
In the case that $Q_{\mathrm{BB}}^{(R)}$ is not invertible, refer to the discussion below Eq.~\eqref{eq:GeneralBoundInvertibile}.
This completes the proof.
\end{proof}
Note that we  assumed   $\hat\rho_\lambda$ to be faithful  to ensure that $\hat R_k$ exists and is unique. The assumption is not necessary for the existence of $\hat R_k$; however, without faithfulness of $\hat\rho_\lambda$, $\hat R_k$ need not be unique.
Also note that in the above proof we do not require $\hat\rho_{\lambda_k}$ to be faithful for all $k=1,\ldots,K$ with $\lambda_k\neq\lambda$.

Let us now compare the derived quantum Barankin matrix
\begin{align}
     (Q_{\mathrm{BB}}^{(R)})_{kl} = \mathrm{Tr}[\hat\rho_{\lambda_k}\hat\rho_{\lambda}^{-1}\hat\rho_{\lambda_l}]  
\end{align}
with the classical Barankin matrix $ (C_{\mathrm{BB}})_{kl}= \int_{D}\mathrm{d}^dx\,\frac{p(\mathbf x\vert\lambda_k)p(\mathbf x\vert\lambda_l)}{p(\mathbf x\vert\lambda)}$. In this form, it becomes evident that our quantum Barankin matrix based on the right division operator is a natural generalization of the classical version.

When $\hat\rho_\lambda $ is faithful, we can write $\hat R_k$ in the eigenbasis $\vert\mathbf n\rangle$ of $\hat\rho_{\lambda} = \sum_{\mathbf n} p_{\mathbf n} \vert\mathbf{n}\rangle\langle\mathbf{n}\vert$ as
\begin{align}\label{eq:RDO}
    \hat R_k = \sum_{\mathbf{n},\mathbf{m}}\frac{\langle\mathbf m\vert   \hat \rho_{\lambda_k}\vert\mathbf n\rangle }{p_{\mathbf m}}\vert \mathbf m \rangle\langle\mathbf n\vert  .
\end{align}

\subsection{Attainability and relation to Quantum Barankin bound based on the symmetric division operator} \label{subsec:Attainability}
Let us discuss how our newly derived quantum Barankin bound based on  the right division operator (RD QBB) compares to the quantum Barankin bound based on the symmetric division operator (SD QBB) proposed in \cite{gessner2023hierarchies}. It was shown that an optimal measurement that attains the SD QBB is given by projecting onto $\hat\Pi_{ x}$, which are projectors onto the eigenstates of $\hat \Omega_{\hat\rho_{\lambda_{\mathrm{pr}}}}(\sum_{k=1}^K \hat\rho_{\lambda_k} a_k)$. Note that when $Q_{\text{BB}}^{(S)}$ is  invertible, which is the case we consider here for simplicity, the optimal choice for $\mathbf{a}$ is $\mathbf{a} = (Q_{\text{BB}}^{(S)})^{-1}\boldsymbol\gamma$. The corresponding probability distribution is $p(x\vert \lambda) = \text{Tr}[\hat \Pi_x \hat \rho_\lambda]$. Furthermore, it was shown that the estimator $\lambda_{\text{est}} (x) = \lambda_{\text{pr}} +  \sum_{k,l=1}^K\frac{p(x\vert\lambda_k)}{p(x\vert\lambda_{\text{pr}})} (Q_{\text{BB}}^{(S)})_{kl}^{-1} \gamma_l$ satisfies the constraints encoded in $\boldsymbol{\gamma}$ and satisfies $\text{MSE}_\lambda[\lambda_{\text{est}}]  \simeq \boldsymbol\gamma^\dag (Q_{\text{BB}}^{(S)})^{-1} \boldsymbol\gamma$, provided the prior guess $\lambda_{\text{pr}}$ is close enough to the true value $\lambda$. Thus, the SD QBB is attainable which leads us to the conclusion:
\begin{align}
    \text{MSE}_\lambda[\lambda_{\text{est}}]   \geq \boldsymbol{\gamma}^\dag (Q_{\text{BB}}^{(S)})^{-1} \boldsymbol{\gamma} \geq \boldsymbol{\gamma}^\dag (Q_{\text{BB}}^{(R)})^{-1} \boldsymbol{\gamma} ,
\end{align}
which is analogous to the QCRB case in Eq.~\eqref{eq:CRBhierarchy}.
This means that the RD QBB is generally unattainable. The SD QBB is therefore more informative. However, as we will see later in Section~\ref{sec:Gaussian}, the RD QBB is often easier to calculate, especially for Gaussian states.

\subsection{Recovering the RLD QCRB}\label{subsec:RecoverRLDCRBgeneral}
Let us now show the close connection between the RD QBB and RLD CRB. We only consider two test points: $(\lambda,\lambda_2)^T$, where $\lambda$ is the true value. In fact, this choice leads to a quantum generalization of the Hammersley–Chapman–Robbins bound \cite{hammersley1950,chapman1951}. We thus find
\begin{align}
    Q_{\text{BB}}^{(R)} = \begin{pmatrix}
        1 & 1 \\
        1 & (Q_{\text{BB}}^{(R)})_{22}
    \end{pmatrix}
\end{align}
from which follows
\begin{align}
    (Q_{\text{BB}}^{(R)})^{-1} = \frac{1}{(Q_{\text{BB}}^{(R)})_{22}-1}\begin{pmatrix}
        (Q_{\text{BB}}^{(R)})_{22} & -1 \\
        -1 & 1
    \end{pmatrix}
\end{align}
for $(Q_{\text{BB}}^{(R)})_{22}\neq 1$.

Now, we set $\gamma_k = \lambda_k-\lambda$ to enforce   unbiasedness at the test points $\lambda_2$ and $\lambda$ for the considered estimators. We thus get $\boldsymbol \gamma = (0,\lambda_2 - \lambda)^T$. This leads to
\begin{align}
    \boldsymbol \gamma ^\dag (Q_{\mathrm{BB}}^R)^{-1} \boldsymbol \gamma &= \frac{(\lambda_2-\lambda)^2}{(Q_{\text{BB}}^{(R)})_{22}-1} = \frac{(\lambda_2-\lambda)^2}{\text{Tr}[\hat R_2^\dag \hat\rho_\lambda \hat R_2]-1}    \\
    &=\frac{(\lambda_2-\lambda)^2}{\text{Tr}[\hat R_{\hat\rho_{\lambda}}^\dag (\hat\rho_{\lambda_2}) \hat\rho_\lambda \hat R_{\hat\rho_{\lambda}} (\hat\rho_{\lambda_2})]-1} = \frac{\Delta\lambda^2}{\text{Tr}[\hat R_{\hat\rho_{\lambda}}^\dag (\hat\rho_{\lambda+\Delta\lambda}) \hat\rho_\lambda \hat R_{\hat\rho_{\lambda}} (\hat\rho_{\lambda+\Delta\lambda})]-1}   ,
\end{align}
where we introduced $\Delta\lambda = \lambda_2-\lambda$.

It is easy to see that the right division operator is linear in its argument and we thus have $\hat R_{\hat \rho_\lambda} (\frac{\hat \rho_{\lambda+\Delta\lambda}-\hat\rho_\lambda}{\Delta\lambda}) = (\hat R_{\hat \rho_\lambda}(\hat\rho_{\lambda+\Delta\lambda})-\hat R_{\hat \rho_\lambda}(\hat\rho_\lambda))/\Delta\lambda$.
With this, consider the following expression:
\begin{align}
     \text{Tr}&\left[ \hat R_{\hat \rho_\lambda}^\dag \left(\frac{\hat \rho_{\lambda+\Delta\lambda}-\hat\rho_\lambda}{\Delta\lambda}  \right) \hat\rho_{\lambda}  \hat R_{\hat \rho_\lambda} \left(\frac{\hat \rho_{\lambda+\Delta\lambda}-\hat\rho_\lambda}{\Delta\lambda}  \right) \right] \\
     &=   \frac{1}{\Delta\lambda^2}\text{Tr}\left[ \left(\hat R_{\hat \rho_\lambda}^\dag \left(\hat \rho_{\lambda+\Delta\lambda} \right) -\hat R_{\hat \rho_\lambda}^\dag \left(\hat \rho_{\lambda } \right)\right) \hat\rho_{\lambda}   \left(\hat R_{\hat \rho_\lambda}  \left(\hat \rho_{\lambda+\Delta\lambda} \right) -\hat R_{\hat \rho_\lambda}  \left(\hat \rho_{\lambda } \right)\right) \right]  \\
     & =\frac{1}{\Delta\lambda^2}\text{Tr}\left[ \left(\hat R_{\hat \rho_\lambda}^\dag \left(\hat \rho_{\lambda+\Delta\lambda} \right)\hat\rho_{\lambda}   -\hat R_{\hat \rho_\lambda}^\dag \left(\hat \rho_{\lambda } \right)\hat\rho_{\lambda}  \right)  \left(\hat R_{\hat \rho_\lambda}  \left(\hat \rho_{\lambda+\Delta\lambda} \right) -\hat R_{\hat \rho_\lambda}  \left(\hat \rho_{\lambda } \right)\right) \right] \\
       & =\frac{1}{\Delta\lambda^2}\text{Tr}\left[  \hat R_{\hat \rho_\lambda}^\dag \left(\hat \rho_{\lambda+\Delta\lambda} \right)\hat\rho_{\lambda}\hat R_{\hat \rho_\lambda}  \left(\hat \rho_{\lambda+\Delta\lambda} \right)   -\hat R_{\hat \rho_\lambda}^\dag \left(\hat \rho_{\lambda } \right)\hat\rho_{\lambda} \hat R_{\hat \rho_\lambda}  \left(\hat \rho_{\lambda+\Delta\lambda} \right)    - \hat R_{\hat \rho_\lambda}^\dag \left(\hat \rho_{\lambda+\Delta\lambda} \right)\hat\rho_{\lambda} \hat R_{\hat \rho_\lambda}  \left(\hat \rho_{\lambda } \right)  +\hat R_{\hat \rho_\lambda}^\dag \left(\hat \rho_{\lambda } \right)\hat\rho_{\lambda} \hat R_{\hat \rho_\lambda}  \left(\hat \rho_{\lambda } \right)      \right] \\
          & =\frac{1}{\Delta\lambda^2}\text{Tr}\left[  \hat R_{\hat \rho_\lambda}^\dag \left(\hat \rho_{\lambda+\Delta\lambda} \right)\hat\rho_{\lambda}\hat R_{\hat \rho_\lambda}  \left(\hat \rho_{\lambda+\Delta\lambda} \right)   -2+1    \right]  
           = \frac{\text{Tr}[\hat R_{\hat\rho_{\lambda}}^\dag (\hat\rho_{\lambda+\Delta\lambda}) \hat\rho_\lambda \hat R_{\hat\rho_{\lambda}} (\hat\rho_{\lambda+\Delta\lambda})]-1}{\Delta\lambda^2} ,
\end{align}
where we used the property $\hat R_{\hat A}(\hat A)=\hat I$ and the definition $\hat B = \hat A\hat R_{\hat A}(\hat B)$ and $\hat B^\dag = \hat R_{A}^\dag (\hat B)\hat A^\dag$ to simplify the above.

It thus follows that
\begin{align}
     \boldsymbol \gamma ^\dag (Q_{\mathrm{BB}}^R)^{-1} \boldsymbol \gamma = \frac{1}{\text{Tr} \left[ \hat R_{\hat \rho_\lambda}^\dag \left(\frac{\hat \rho_{\lambda+\Delta\lambda}-\hat\rho_\lambda}{\Delta\lambda}  \right) \hat\rho_{\lambda}  \hat R_{\hat \rho_\lambda} \left(\frac{\hat \rho_{\lambda+\Delta\lambda}-\hat\rho_\lambda}{\Delta\lambda}  \right) \right] } .
\end{align}
We now can take the limit $\Delta\lambda\rightarrow 0$, and under suitable regularity conditions, we find 
\begin{align}
   \mathrm{MSE}_\lambda[\lambda_{\mathrm{est}}] \geq \lim_{\Delta\lambda\rightarrow 0} \boldsymbol \gamma ^\dag (Q_{\mathrm{BB}}^{(R)})^{-1} \boldsymbol \gamma =   \frac{1}{\text{Tr} \left[ \hat R_{\hat \rho_\lambda}^\dag \left( \partial_{\tilde\lambda }\hat\rho_{\tilde\lambda}\vert_{\tilde\lambda=\lambda} \right) \hat\rho_{\lambda}  \hat R_{\hat \rho_\lambda} \left( \partial_{\tilde\lambda }\hat\rho_{\tilde\lambda}\vert_{\tilde\lambda=\lambda}   \right) \right] } = \frac{1}{\mathcal{F}^{(R)}(\hat\rho_\lambda)} .
\end{align}
Thus, in the limit $\Delta\lambda\rightarrow 0$ we recover the RLD QCRB from the RD QBB.

\section{Gaussian-state formulas for Quantum Barankin bound}\label{sec:Gaussian}
Above we have derived the general quantum Barankin bound based on the right division operator that applies to general infinite and finite dimensional systems. Let us now specialize the general results to the class of Gaussian states. Such Gaussian state formulas exist for various versions of the QCRBs~\cite{safranek2018a,safranek2015,gao2014bounds,bakmou2020multiparameter,monras2013,jiang2014,sorelli2024a}, but formulas for the SD QBB are unknown. Here we show that Gaussian state formulas for the RD QBB are readily derivable.

\subsection{Preliminaries on Gaussian states}
Let us first  give a brief review of the bosonic Gaussian state formalism. We follow the notation of Ref.~\cite{seshadreesan2018renyi}. Many useful properties that were stated or proved in Ref.~\cite{seshadreesan2018renyi} are listed in Appendix~\ref{app:GaussianRelations}.

We consider $n$ bosonic modes with the $2n$ quadrature operators
\begin{align}
   \hat{\mathbf{x}} =  (\hat x_1,\hat p_1,\hat x_2,\hat p_2,\ldots, \hat x_n,\hat p_n)^T
\end{align}
 that satisfy the commutation relations
\begin{align}
    [\hat x_i,\hat x_j] = i\Omega_{ij}, \quad \text{where }  \Omega = I_n\otimes \begin{pmatrix}
        0 & 1\\
        -1 & 0
    \end{pmatrix}
\end{align}
is the symplectic form and $I_n$ is the $n\times n$ identity matrix.

 A faithful  Gaussian state of $n$ modes can be written as
 \begin{align}
       \hat \rho =\frac{1}{Z_{\hat\rho}   } \hat D(-\mathbf s_{\hat\rho}  )e^{\hat H_{\hat\rho}  } \hat D(\mathbf s_{\hat\rho}  ), \label{eq:GaussianStateDef}
 \end{align}
 where
\begin{align}
    \hat H_{\hat\rho}    = -\frac{1}{2}\hat {\mathbf x}^T H_{\hat\rho}   \hat {\mathbf x}
\end{align} 
is the Hamiltonian, and $H_{\hat\rho} >0$ is the $2n\times 2n$  positive-definite real symmetric Hamiltonian matrix, and 
\begin{align}
    \hat D (\mathbf{s}) = \exp [ \mathbf{s}^T i \Omega \hat{ \mathbf{x}}] = \exp [ - \hat {\mathbf{x}} ^T i \Omega  \mathbf{s}] \label{eq:defDisp}
\end{align}
is the displacement operator with $\mathbf s\in \mathds{R}^{2n}$ and
\begin{align}
    Z_{\hat\rho} = \sqrt{\mathrm{Det}\left( \frac{V_{\hat\rho}+i\Omega}{2} \right)} .
\end{align}
We refer to $\mathbf s_{\hat\rho}\in \mathds{R}^{2n}$ as the mean vector and $V_{\hat\rho}$ as the covariance matrix, which is related to $H_{\hat\rho}$ as:
\begin{align}
 V_{\hat\rho} &=  \coth (i\Omega H_{\hat\rho} /2) i\Omega\label{eq:VofH}, \\
    H_{\hat\rho}  &=  2i\Omega \mathrm{arccoth}( V_{\hat\rho} i\Omega) \label{eq:HofV} .
\end{align}
Note that $H_{\hat\rho}$ real, symmetric and positive definite implies that $V_{\hat\rho}$ is a real and symmetric $2n\times2n$ matrix and that $V_{\hat\rho}+i\Omega>0$ is satisfied. Pure states, that are not faithful, can be included by allowing $V_{\hat\rho}+i\Omega\geq0$.
We also note that for (normalized) Gaussian states  we have $\mathbf s_{\hat\rho} = \mathrm{Tr}[\hat{ \mathbf x}\hat\rho ]$ and $ (V_{\hat\rho})_{ij} =   \mathrm{Tr}[\{\hat{  x}_i- s_{\hat\rho,i},\hat{  x}_j- s_{\hat\rho,j} \}  \hat\rho ]$.

Let us also briefly discuss powers of Gaussian states.
The (unnormalized) power $\hat \rho^\alpha$ of the  Gaussian state $\hat\rho$ with $\alpha>0$ has covariance matrix
\begin{align}
    V_{\hat\rho^\alpha} = \coth(\alpha \mathrm{arccoth}(V_{\hat\rho}i\Omega))i\Omega = \frac{(I+(V_{\hat\rho}i\Omega)^{-1})^\alpha +(I-(V_{\hat\rho}i\Omega)^{-1})^\alpha }{(I+(V_{\hat\rho}i\Omega)^{-1})^\alpha -(I-(V_{\hat\rho}i\Omega)^{-1})^\alpha }i\Omega   ,
\end{align}
as was shown in Ref.~\cite{seshadreesan2018renyi} by using Eqs.~\eqref{eq:VofH}-\eqref{eq:HofV}. Note that according to our definitions, we have $V_{\hat\rho^\alpha} = V_{\hat\rho^\alpha/\mathrm{Tr}[\hat\rho^\alpha]}$, i.e., our definition gives the same covariance matrix for normalized and unnormalized Gaussian states.
Also note that that the above can be extended to negative powers and we have $ V_{\hat\rho^{-\alpha}} =  -V_{\hat\rho^{\alpha}}$.

We further note that the square $\hat \rho^2$ of a Gaussian $\hat\rho$  is an (unnormalized) Gaussian state and has covariance matrix (see Corollary 3 of Ref.~\cite{seshadreesan2018renyi})
\begin{align}
    V_{\hat\rho^2} = \frac{1}{2} \left( V_{\hat\rho} + \Omega^T V_{\hat\rho}^{-1}\Omega\right)\label{eq:SquareCov}
\end{align}
and the (unnormalized) Gaussian state $\hat \rho^{1/2}$ has covariance matrix (see Corollary 4 of Ref.~\cite{seshadreesan2018renyi})
\begin{align}
    V_{\hat\rho^{1/2}} = \left( \sqrt{I+(V_{\hat\rho}\Omega)^{-2}} +  I \right)   V_{\hat\rho} \:.\label{eq:SquareRootCov}
\end{align}

We may also represent continuous variable quantum states, and more generally, trace-class operators $\hat O$,  with their characteristic function
\begin{align}
  \chi_{\hat O}(\mathbf z)  = \mathrm{Tr}[\hat O \hat D(-\mathbf z)]
\end{align}
as they can be decomposed as $ \hat O = \frac{1}{(2\pi)^n} \int \mathrm{d}^{2n} z \,   \chi_{\hat O}(\mathbf z) \hat D(\mathbf z)$.
We refer to trace-class operators as Gaussian trace-class operators if their characteristic function is Gaussian, i.e., if they are of the form
\begin{align}
    \chi_{\hat O}(\mathbf  z) = c_{\hat O} e^{-\frac{1}{4}\mathbf  z^T \Omega^T V_{\hat O}\Omega \mathbf  z}e^{i\mathbf z ^T\Omega^T \mathbf s_{\hat O}},
\end{align}
where $\mathrm{Tr}[\hat O] = c_{\hat O} \in \mathds{C}$, $\mathbf s_{\hat O}\in\mathds{C}^{2n}$  and $V_{\hat O}$ is a complex symmetric $2n\times 2n$ matrix. We will also refer to $\mathbf s_{\hat O}$ and $V_{\hat O}$ as (generalized) mean vector and covariance matrix, respectively, of the Gaussian operator $\hat O$.   Note that a Gaussian state $ \hat O =\hat\rho $ is a special case of Gaussian trace-class operator with $c_{\hat O} = 1$ and mean vector $\mathbf s_{\hat O}$ and covariance matrix $V_{\hat O}$  with properties as described above.  Another example of such a Gaussian trace-class operator is the product of two Gaussian states $\hat O = \hat\rho_1\hat\rho_2$ whose characteristic function we calculate in Appendix~\ref{app:PropertiesProduc} in Theorem~\ref{theorem:CharacteristicProd}.

\subsection{Quantum Barankin bound for Gaussian states}
In Section~\ref{sec:GeneralResult}, we derived the quantum Barankin bound based on the right division operator for general infinite dimensional systems. We derived general sufficient conditions for the existence and a general form of the quantum Barankin matrix. Now, we  narrow down the set of quantum states to the class of Gaussian states, and we  reformulate Theorem~\ref{theorem:1}. This allows us to express the sufficient existence conditions and Barankin matrix elements entirely in terms of mean vectors and covariance matrices.

\begin{theorem}
    Let $\{ \hat\rho_{\tilde\lambda}\}_{\tilde\lambda\in\Lambda}$ be a family of Gaussian states. 
    Let $\hat\rho_{\lambda}  =: \hat\rho$ be the actual Gaussian state, where $\lambda\in\Lambda$ is the true fixed value to be estimated.  Let $\hat\rho$ be faithful with mean vector $\mathbf s_{\hat \rho }$ and covariance matrix $V_{\hat \rho }$. Let $\{ \lambda_k\}_{k=1}^K$ be a set of test points with $\lambda_k\in\Lambda$  and $\hat\rho_{\lambda_k} =:\hat\rho_k$ a Gaussian state with mean vector $\mathbf s_{\hat\rho_{k}}$ and covariance matrix $V_{\hat\rho_{k}}$ so that
    \begin{align}
          2 V_{\hat\rho }- V_{\hat\rho_k} - \Omega^T V_{\hat\rho_k}^{-1}\Omega >0  \quad \forall k\in\{1,\ldots ,K\} . \label{eq:CovConditionBB}
    \end{align}
     It then follows that for every pair of POVM  $\{\Pi_{\mathbf x}\}_{\mathbf x}$ and estimator $\lambda_{\mathrm{est}}$ that satisfy the bias constraints 
     \begin{align}
     \int_D \mathrm{d}^d x \, \mathrm 
    {Tr}[\hat \rho_{\lambda_k}\hat \Pi_{\mathbf x}] \lambda_{\mathrm{est}}(\mathbf x) = \lambda_k+b_{\lambda_k}, 
    \end{align}
    we have
    \begin{align}
         \mathrm{MSE}_\lambda[\lambda_{\mathrm{est}}] \geq \sup_{\mathbf{a}\neq 0} \frac{\vert\mathbf{a}^\dag \boldsymbol{\gamma}\vert^2}{\mathbf{a}^\dag Q_{\mathrm{BB}}^{(R)} \mathbf{a}} 
    \end{align}
    with $\mathbf{a}\in \mathds{C}^K$, $\gamma_k = b_{\lambda_k }+\lambda_k-\lambda$.
     If additionally $Q_{\mathrm{BB}}^{(R)}$ is invertible~\cite{BarankinMatrix}, we have
    \begin{align}
         \mathrm{MSE}_\lambda[\lambda_{\mathrm{est}}] \geq \boldsymbol \gamma^\dag (Q_{\mathrm{BB}}^{(R)})^{-1} \boldsymbol \gamma .
    \end{align}
    
    The elements of the quantum Barankin matrix $Q_{\mathrm{BB}}^{(R)}$  are given by
    \begin{align}
        ( Q_{\mathrm{BB}}^{(R)})_{kl} &= \mathrm{Tr}[\hat \rho_k\hat\rho^{-1}  \hat \rho_l] \\
   &= c_{\hat\rho_k\hat\rho^{-1/2}}c_{\hat\rho^{-1/2}\hat\rho_l}\frac{\exp\left(-(\mathbf s_{\hat\rho_k\hat\rho^{-1/2}}-\mathbf s_{\hat\rho^{-1/2}  \hat\rho_l})^T (V_{\hat\rho_k\hat\rho^{-1/2}}+V_{\hat\rho^{-1/2}  \hat\rho_l})^{-1}(\mathbf s _{\hat\rho_k\hat\rho^{-1/2}}-\mathbf s _{\hat\rho^{-1/2}  \hat\rho_l})\right)  }{\exp\left(\frac{1}{2}\mathrm{Tr}[\log(\frac{V_{\hat\rho_k\hat\rho^{-1/2}}+V_{\hat\rho^{-1/2}  \hat\rho_l}}{2})]\right)}   ,\label{eq:QBBgaussianFormaula}
    \end{align}
    where $\log$ is the principal matrix logarithm.   Alternatively, we can write
    $\exp\left(\frac{1}{2}\mathrm{Tr}[\log(\frac{V_{\hat\rho_k\hat\rho^{-1/2}}+V_{\hat\rho^{-1/2}  \hat\rho_l}}{2})]\right)= \prod_{i=1}^{2n}\sqrt{ \mu_i}=:\sqrt{\mathrm{Det}\left(\frac{V_{\hat\rho_k\hat\rho^{-1/2}}+V_{\hat\rho^{-1/2}  \hat\rho_l}}{2}\right)}$ where each $\sqrt{ \mu_i}$ is the principal square root of the eigenvalues $\mu_i$  of $ \frac{V_{\hat\rho_k\hat\rho^{-1/2}}+V_{\hat\rho^{-1/2}  \hat\rho_l}}{2}$.
    We further have
    \begin{align}
    c_{\hat\rho_k\hat\rho^{-1/2}} &=     \mathrm{Det}\!\left[\tfrac{V_{\hat\rho}+i\Omega}{2}\right]^{1/4}\sqrt{\frac{\mathrm{Det}[\frac{V_{\hat\rho^{1/2}}+i\Omega}{2}]}{ \mathrm{Det}[\frac{V_{\hat\rho^{1/2}}-V_{\hat\rho_k}}{2}]}}e^{- \delta\mathbf s_k^T (V_{\hat\rho_k}-V_{\hat\rho^{1/2}})^{-1} \delta\mathbf s_k} \, \label{def:c}\\
   \mathbf s_{\hat\rho_k\hat\rho^{-1/2}} &=     \mathbf s_{\hat\rho}   + (-V_{\hat\rho^{1/2}}+i\Omega) (-V_{\hat\rho^{1/2}}+V_{\hat\rho_k})^{-1}   \delta \mathbf s_k \label{def:s} \\
   V_{\hat\rho_k\hat\rho^{-1/2}} &=  -i\Omega + (-V_{\hat\rho^{1/2}}+i\Omega) ( V_{\hat\rho_k}-V_{\hat\rho^{1/2}})^{-1} (V_{\hat\rho_k}+i\Omega) ,\label{def:V}
\end{align}
where  $\delta\mathbf{s}_k := \mathbf s_{\hat\rho_k}- \mathbf s_{\hat\rho}$,  and we further have $c_{\hat\rho^{-1/2}\hat\rho_k}  = c_{\hat\rho_k \hat\rho^{-1/2}}^*$,  $ \mathbf s_{\hat\rho^{-1/2}\hat\rho_k}  = \mathbf s_{\hat\rho_k\hat\rho^{-1/2}}^* 
    $ and $ V_{\hat\rho^{-1/2}\hat\rho_k}  = V_{\hat\rho_k\hat\rho^{-1/2}}^*$
\end{theorem}

\begin{proof}
    As a starting point we use Theorem~\ref{theorem:1}, which is the general statement of the theorem. There it was proven that $\hat\rho_{\lambda_k}\hat\rho_\lambda^{-1/2}$ being trace class is a sufficient condition for the existence of the quantum Barankin matrix elements. Now, for Gaussian states this sufficient condition can be translated to a sufficient condition on the covariance matrices. First, note that 
    $
        \hat \rho = \frac{1}{Z_{\hat\rho}} \hat D(-\mathbf s_{\hat\rho}) e^{\hat H_{\hat\rho}} \hat D(\mathbf s_{\hat\rho}) $ and $\hat \rho_k = \frac{1}{Z_{\hat\rho_k}} \hat D(-\mathbf s_{\hat\rho_k}) e^{\hat H_{\hat\rho_k}} \hat D(\mathbf s_{\hat\rho_k})$. Furthermore, we have $\hat \rho^{-1/2} = \frac{1}{Z_{\hat\rho}^{-1/2}} \hat D(-\mathbf s_{\hat\rho}) e^{-\hat H_{\hat\rho}/2} \hat D(\mathbf s_{\hat\rho}) $.
        In Appendix~\ref{app:PropertiesProduc} Theorem~\ref{theorem:CharacteristicProd} in Eqs.~\eqref{eq:C1}-\eqref{eq:C4} in combination with the comments from Corollary~\ref{corollary:CharacteristicInverse} it was shown that we can rewrite the product of exponential operators as
        \begin{align}
             \hat\rho_k \hat\rho^{-1/2} = c \hat D(\mathbf r_1)  e^{\hat H_{\hat\rho_k}} e^{-\hat H_{\hat\rho}/2}\hat D(\mathbf r_2) ,
        \end{align}
        where $c$ is an irrelevant complex constant and $\mathbf r_1,\mathbf r_2$ are irrelevant real vectors. Now, in Theorem~\ref{theorem:TraceClassCondition}, it was shown that the product $\hat \rho_1 \hat \rho_2^{-1}$ of two (unnormalized) Gaussian states $\hat\rho_1,\hat\rho_2$ is trace class if $V_{\hat\rho_1^2} - V_{\hat \rho_2^2}> 0 $. In our case, we have $\hat\rho_1 =e^{\hat H_{\hat\rho_k}}$ and $\hat\rho_2 =(e^{\hat H_{\hat\rho}})^{1/2}$, and thus, $e^{\hat H_{\hat\rho_k}}e^{-\hat H_{\hat\rho}/2}$ is trace class if $V_{(\hat\rho^{1/2})^2} - V_{\hat \rho_k^2} = V_{\hat\rho } - V_{\hat \rho_k^2}> 0 $, which is equivalent to the condition given in Eq.~\eqref{eq:CovConditionBB}, which we can show by using Eq.~\eqref{eq:SquareCov}.

 The next step is to find analytic expressions of $Q_{\text{BB}}^{(R)}$ in terms of the mean and covariance. Note that $(Q_{\text{BB}}^{(R)})_{kl} = \mathrm{Tr}[\hat \rho_k\hat\rho^{-1}  \hat \rho_l] = \mathrm{Tr}[(\hat \rho_k\hat\rho^{-1/2}) (\hat\rho^{-1/2}  \hat \rho_l)]$. To calculate this trace, let us first calculate the characteristic functions of $\hat \rho_k\hat\rho^{-1/2}$ and $\hat\rho^{-1/2}  \hat \rho_l$. We start with $\hat \rho_k\hat\rho^{-1/2}$:
 \begin{align}
         \chi_{\hat \rho_k\hat\rho^{-1/2}} (\mathbf z) &= \mathrm{Tr}[\hat \rho_k\hat\rho^{-1/2}\hat D(-\mathbf z)]\\
           &= \mathrm{Tr}\left[\left(\frac{1}{Z_{\hat\rho_k}} \hat D(-\mathbf s_{\hat\rho_k}) e^{\hat H_{\hat\rho_k}} \hat D(\mathbf s_{\hat\rho_k})\right)\left(\frac{1}{Z_{\hat\rho}^{-1/2}} \hat D(-\mathbf s_{\hat\rho}) e^{-\hat H_{\hat\rho}/2} \hat D(\mathbf s_{\hat\rho}) \right)\hat D(-\mathbf z)\right]\\
            &=  \frac{Z_{\hat\rho}^{1/2}}{Z_{\hat\rho_k}}\mathrm{Tr}\left[  \hat D(-\mathbf s_{\hat\rho_k}) e^{\hat H_{\hat\rho_k}} \hat D(\mathbf s_{\hat\rho_k})  \hat D(-\mathbf s_{\hat\rho}) e^{-\hat H_{\hat\rho}/2} \hat D(\mathbf s_{\hat\rho}) \hat D(-\mathbf z)\right].
 \end{align}
Now, using Corollary~\ref{corollary:CharacteristicInverse}, which follows from Theorem~\ref{theorem:CharacteristicProd} in Appendix~\ref{app:PropertiesProduc}, we find that the characteristic function is a complex Gaussian
\begin{align}
    \chi_{\hat \rho_k\hat\rho^{-1/2}} (\mathbf z)  = c_{\hat\rho_k\hat\rho^{-1/2}} e^{i\mathbf{z}^T \Omega^T \mathbf s_{\hat\rho_k\hat\rho^{-1/2}}} e^{- \frac{1}{4}\mathbf z^T \Omega^T V_{\hat\rho_k\hat\rho^{-1/2}}  \Omega \mathbf z } ,\label{eq:CharacteristicFunctionProduct}
\end{align}
where $c_{\hat\rho_k\hat\rho^{-1/2}}, \mathbf s_{\hat\rho_k\hat\rho^{-1/2}}, V_{\hat\rho_k\hat\rho^{-1/2}}$ are defined in Eqs.~\eqref{def:c}-\eqref{def:V}. Note that the inverse $(V_{\hat\rho_k}-V_{\hat{\rho}^{1/2}})^{-1}$ exists, which follows from the fact that $V_{\hat\rho}-V_{\hat\rho_k^2}>0 $ (by assumption) and Corollary~\ref{corollary:Invertibility}.

To calculate the characteristic function of $\hat\rho^{-1/2}\hat\rho_k$, we  note that  $\text{Tr}[A] = \text{Tr}[A^\dag]^*$ for $\hat A$ trace class. Thus, we have 
\begin{align}
    \chi_{\hat\rho^{-1/2}  \hat\rho_k}(\mathbf z) &= \text{Tr}[\hat \rho^{-1/2}\hat \rho_k \hat D(-\mathbf z)] = \text{Tr}[(\hat \rho^{-1/2}\hat \rho_k \hat D(-\mathbf z))^\dag]^*\\
    &= \text{Tr}[\hat D(-\mathbf z)^\dag (\hat \rho^{-1/2}\hat \rho_l)^\dag]^* = \text{Tr}[\hat D(-(-\mathbf z)) \hat \rho_k\hat \rho^{-1/2}]^* \\
    &= \text{Tr}[ \hat \rho_k\hat \rho^{-1/2}\hat D(-(-\mathbf z))]^* \\
    &= \chi_{\hat\rho_k \hat \rho^{-1/2}} (-\mathbf z)^* .
\end{align}
Now, $\chi_{\rho_k  \rho^{-1/2}} (\mathbf z)$ is given in Eq.~\eqref{eq:CharacteristicFunctionProduct}. We thus obtain
\begin{align}
   \chi_{\hat\rho^{-1/2}  \hat\rho_k}(\mathbf z) =   ( \chi_{\hat \rho_k\hat\rho^{-1/2}} (- \mathbf z) )^* &= \left(c_{\hat\rho_k\hat\rho^{-1/2}} e^{i(-\mathbf{z})^T \Omega^T \mathbf s_{\hat\rho_k\hat\rho^{-1/2}}} e^{- \frac{1}{4}  (-\mathbf z)^T \Omega^T V_{\hat\rho_k\hat\rho^{-1/2}}  \Omega (-\mathbf z )} \right)^*\\
    &= \left(c_{\hat\rho_k\hat\rho^{-1/2}}\right)^* \left( e^{i(-\mathbf{z})^T \Omega^T \mathbf s_{\hat\rho_k\hat\rho^{-1/2}}}\right)^* \left( e^{- \frac{1}{4}   \mathbf z ^T \Omega^T V_{\hat\rho_k\hat\rho^{-1/2}}  \Omega  \mathbf z  } \right)^* \\
    &= c_{\hat\rho_k\hat\rho^{-1/2}} ^*   e^{(-i)(-\mathbf{z})^T \Omega^T \mathbf s_{\hat\rho_k\hat\rho^{-1/2}}^*}    e^{- \frac{1}{4}   \mathbf z ^T \Omega^T V_{\hat\rho_k\hat\rho^{-1/2}}^*  \Omega  \mathbf z  }   \\
    &= c_{\hat\rho_k\hat\rho^{-1/2}} ^*   e^{ i \mathbf{z} ^T \Omega^T \mathbf s_{\hat\rho_k\hat\rho^{-1/2}}^*}    e^{- \frac{1}{4}   \mathbf z ^T \Omega^T V_{\hat\rho_k\hat\rho^{-1/2}}^*  \Omega  \mathbf z  }   .
\end{align}
Now, we find that  
\begin{align}
    \chi_{\hat\rho^{-1/2}  \hat\rho_k}(\mathbf z) = c_{\hat\rho^{-1/2}\hat\rho_k}    e^{ i \mathbf{z} ^T \Omega^T \mathbf s_{\hat\rho^{-1/2}\hat\rho_k}}    e^{- \frac{1}{4}   \mathbf z ^T \Omega^T V_{\hat\rho^{-1/2}\hat\rho_k}  \Omega  \mathbf z  }   
\end{align}
is also Gaussian with
\begin{align}
    c_{\hat\rho^{-1/2}\hat\rho_k} &= c_{\hat\rho_k \hat\rho^{-1/2}}^*\\
    \mathbf s_{\hat\rho^{-1/2}\hat\rho_k} &= \mathbf s_{\hat\rho_k\hat\rho^{-1/2}}^* \\
     V_{\hat\rho^{-1/2}\hat\rho_k} &= V_{\hat\rho_k\hat\rho^{-1/2}}^* .
\end{align}
Thus, we have two trace class operators with complex Gaussian characteristic function. The trace of these can be calculated as
\begin{align}
    \mathrm{Tr}[(\hat \rho_k\hat\rho^{-1/2}) (\hat\rho^{-1/2}  \hat \rho_l)] &= \frac{1}{(2\pi)^n(2\pi)^n}\int\mathrm d^{2n}z_1 \int\mathrm d^{2n}z_2  \, \chi_{\hat \rho_k\hat\rho^{-1/2}} (\mathbf z_1) \chi_{\hat\rho^{-1/2}  \hat \rho_l} (\mathbf{z_2}) \underbrace{\mathrm{Tr}[\hat D(\mathbf z_1) \hat D(\mathbf z_2)]}_{=(2\pi)^n\delta^{2n}(\mathbf z_1+ \mathbf z_2)} 
    \\
    &= \frac{1}{(2\pi)^n}\int\mathrm d^{2n}z    \, \chi_{\hat \rho_k\hat\rho^{-1/2}} (\mathbf z) \chi_{\hat\rho^{-1/2}  \hat \rho_l} (-\mathbf{z })  \\
    &= \frac{c_{\hat \rho_k\hat\rho^{-1/2}} c_{\hat\rho^{-1/2}  \hat \rho_l}}{(2\pi)^n} \int\mathrm d^{2n}z\, e^{i\mathbf{z}^T \Omega^T (\mathbf s_{\hat\rho_k\hat\rho^{-1/2}} - \mathbf s_{\hat\rho^{-1/2}\hat\rho_l}) } e^{- \frac{1}{4}\mathbf z^T \Omega^T (V_{\hat\rho_k\hat\rho^{-1/2}} +V_{\hat\rho^{-1/2}\hat\rho_l} )  \Omega \mathbf z }  \\
    &= \frac{c_{\hat \rho_k\hat\rho^{-1/2}} c_{\hat\rho^{-1/2}  \hat \rho_l}}{(2\pi)^n} \frac{\pi^n}{\exp(\frac{1}{2}\mathrm{Tr}[\log (\frac{1}{4}\Omega^T (V_{\hat\rho_k\hat\rho^{-1/2}} +V_{\hat\rho^{-1/2}\hat\rho_l} )  \Omega)])}\nonumber \\
    &\quad  \times e^{\frac{1}{4} (i\Omega^T (\mathbf s_{\hat\rho_k\hat\rho^{-1/2}} - \mathbf s_{\hat\rho^{-1/2}\hat\rho_l}) )^T (\frac{1}{4}\Omega^T (V_{\hat\rho_k\hat\rho^{-1/2}} +V_{\hat\rho^{-1/2}\hat\rho_l} )  \Omega)^{-1} (i\Omega^T (\mathbf s_{\hat\rho_k\hat\rho^{-1/2}} - \mathbf s_{\hat\rho^{-1/2}\hat\rho_l}) )}
\end{align}
Here we used $\int_{\mathds R^{2n}}\mathrm d^{2n} z\, e^{-\mathbf z^T A \mathbf z+ \mathbf z^T \mathbf b} = \frac{\pi^n}{\exp(\frac{1}{2}\mathrm{Tr}[\log A])} e^{\frac{1}{4}\mathbf b^T A^{-1}\mathbf b}$ where $A$ is complex symmetric with $\text{Re}[A]>0$ and $\log$ is the principal matrix logarithm. Note that $\exp(\frac{1}{2}\mathrm{Tr}[\log A]) = \prod_{i=1}^{2n} \sqrt{\mu_i}=: \sqrt{\mathrm{Det}[A]}$ where $\sqrt{\mu_i}$ is the principal square root of the eigenvalues $\mu_i$ of $A$.
For our case, we have
$A = \frac{1}{4}\Omega^T (V_{\hat\rho_k\hat\rho^{-1/2}} +V_{\hat\rho^{-1/2}\hat\rho_l} )  \Omega$ and $\mathbf b= i\Omega^T (\mathbf s_{\hat\rho_k\hat\rho^{-1/2}} - \mathbf s_{\hat\rho^{-1/2}\hat\rho_l})$. 
Note that due to  $\hat \rho_k\hat\rho^{-1/2}$ and$ \hat\rho^{-1/2}  \hat \rho_l$ being trace class operators it follows that $\text{Re}[V_{\hat\rho_k\hat\rho^{-1/2}} ],\text{Re}[V_{\hat\rho^{-1/2}\hat\rho_l} ]>0$\footnote{If we did not have $\text{Re}[V_{\hat\rho_k\hat\rho^{-1/2}} ],\text{Re}[V_{\hat\rho^{-1/2}\hat\rho_l} ]>0$, the characteristic functions $  \chi_{\hat\rho^{-1/2}  \hat\rho_k}(\mathbf z)$ and $  \chi_{\hat\rho_l\hat\rho^{-1/2}  }(\mathbf z)$ would be unbounded, which is a contradiction to the initial assumption that $\hat \rho_k\hat\rho^{-1/2}$ and $\hat\rho^{-1/2}  \hat \rho_l$ are trace class}. From this follows that $\text{Re}[A]>0$. Further, note that we can simplify $\exp(\frac{1}{2}\mathrm{Tr}[\log(\frac{1}{4}\Omega^T (V_{\hat\rho_k\hat\rho^{-1/2}} +V_{\hat\rho^{-1/2}\hat\rho_l} )  \Omega)]) = \sqrt{\mathrm{Det}[\frac{1}{4}\Omega^T (V_{\hat\rho_k\hat\rho^{-1/2}} +V_{\hat\rho^{-1/2}\hat\rho_l} )  \Omega]} = \frac{1}{2^{n}}\sqrt{\mathrm{Det}[\frac{1}{2}  (V_{\hat\rho_k\hat\rho^{-1/2}} +V_{\hat\rho^{-1/2}\hat\rho_l} )   ]} = \frac{1}{2^n} \exp(\frac{1}{2}\mathrm{Tr}[\log(\frac{1}{2}  (V_{\hat\rho_k\hat\rho^{-1/2}} +V_{\hat\rho^{-1/2}\hat\rho_l} )   )])$. Now, further simplifications are made eliminating the remaining $\Omega$ matrices which leads us to the final result in Eq.~\eqref{eq:QBBgaussianFormaula}.    
\end{proof}

We thus have successfully restated our general Theorem~\ref{theorem:1} for Gaussian states and have expressed the sufficient existence conditions and the quantum Barankin matrix elements in terms of mean vectors $\mathbf s_{\hat\rho}$, $\hat s_{\hat\rho_k}$ and covariance matrices $V_{\hat\rho}$, $V_{\hat \rho_k}$.

Note that the observations made for the general case in Section~\ref{sec:GeneralResult} also apply to this special case of Gaussian states. Particularly,  the recovery of the RLD QCRB shown for the general case in \ref{subsec:RecoverRLDCRBgeneral} could now be directly done on the level of mean vectors and covariance matrices. The explicit formulas we found were rather cumbersome and we were unable to further simplify them. However, in Section~\ref{sec:Applications}, we numerically verified the recovery of the RLD QCRB from our Gaussian state formulas for the quantum Barankin matrix.

We also note that we have calculated the classical Barankin matrix for classical Gaussian probability distributions in Appendix~\ref{app:ClassicalBarankinGaussian}. Comparing this to its quantum version reveals an interesting structural similarity.

\section{Applications}\label{sec:Applications}
In this section, we evaluate the right quantum Barankin bound (RD QBB) for several examples of interest, imposing the condition of unbiasedness at the selected test points.
First, we derive an analytical expression for the RD QBB matrix elements in the context of phase estimation with $N$ qubits. For this specific task, an efficient method to compute the symmetric Barankin bound (SD QBB) was recently developed \cite{gessner2023hierarchies}, providing a natural benchmark for our results. By comparing the RD QBB with both the SD QBB and the symmetric Cramér-Rao Bound (SLD QCRB), we observe that the RD QBB closely follows the scaling of the symmetric bound as a function of $N$. 
Secondly, we use the results for Gaussian states, namely Eq.~\eqref{eq:QBBgaussianFormaula}, to compute the RD QBB for several estimation tasks. To demonstrate the versatility of our approach, we consider three different scenarios: (i) phase sensing, where the parameter is encoded in the first moment (displacement) of the state; (ii) thermometry, where the temperature is encoded exclusively in the second moment (covariance matrix) \cite{cenni2022thermometry}; and (iii) a generalized case where both the displacement and the covariance matrix are parameter-dependent. In each of these applications, we highlight the emergence of threshold behavior \cite{mcaulay1971a, knockaert1997barankin}. This refers to the abrupt transition in estimation performance occurring at low signal-to-noise ratios, where the MSE increases significantly above the SLD-CRB due to global ambiguities in the parameter encoding. Because the RD QBB incorporates information from test points across the entire parameter space, it successfully captures these non-local effects, providing a tighter and more realistic precision limit in the non-asymptotic regime than local bounds like the CRB.

\subsection{Qubits}
We now apply our formalism to the case of quantum phase estimation with $N$ qubits in the state $\hat{\rho}_\lambda^{\otimes N}$ where $\hat{\rho}_\lambda = \hat{U}_\lambda\hat{\rho}_0 \hat{U}_\lambda^\dagger$ \cite{toth2014,giovannetti2004quantum}. The unitary transformation is defined as $\hat{U}_\lambda = e^{-i\frac{\lambda}{2}\mathbf{n}\cdot\boldsymbol{\sigma}}$, with $\mathbf{n}=(n_x,n_y,n_z)$ being the rotation axis and $\boldsymbol{\sigma}=(\sigma_x,\sigma_y,\sigma_z)$ the vector of Pauli matrices. The single-qubit state can be conveniently expressed in the Bloch representation as $\hat{\rho}_\lambda = \frac{1}{2} \left( \mathbb{I} + \mathbf{r}_\lambda \cdot \boldsymbol{\sigma} \right)$, where the evolved Bloch vector is given by 
\begin{align}
    \mathbf{r}_\lambda = (\cos \lambda)(\mathbf{r}_0 - (\mathbf{n} \cdot \mathbf{r}_0)\mathbf{n}) + (\mathbf{n} \cdot \mathbf{r}_0)\mathbf{n} + (\sin \lambda)\mathbf{n} \times \mathbf{r}_0.
\end{align}
The eigendecomposition of $\hat{\rho}^{\otimes N}_\lambda$ is determined by the single-qubit eigenvectors $|\pm\rangle\langle\pm| = \frac{1}{2} \left( \mathbb{I} \pm \frac{\mathbf{r}_\lambda}{|\mathbf{r}_\lambda|} \cdot \boldsymbol{\sigma} \right)$ and their corresponding eigenvalues $p_\pm = \frac{1 \pm |\mathbf{r}_\lambda|}{2}$ \cite{gessner2023hierarchies}. Therefore, using Eq.~\eqref{eq:RDO} and Eq.~\eqref{eq:RBaB}, the matrix elements of the RD QBB are,
\begin{align}\label{eq:quBa}
     (Q_{\mathrm{BB}}^{(R)})_{kl} = \sum_{\mathbf{m}}\frac{\langle\mathbf m\vert \hat \rho_{\lambda_l}^{\otimes N}\hat\rho_{\lambda_k}^{\otimes N}\vert\mathbf m\rangle }{p_{\mathbf m}} = \sum_{\mathbf{m}}\prod^{N}_i\frac{\langle m_i\vert \hat \rho_{\lambda_l}\hat\rho_{\lambda_k}\vert m_i\rangle }{p_{m_i}} = \prod^{N}_i \left[\frac{\langle +\vert \hat \rho_{\lambda_l}\hat\rho_{\lambda_k}\vert +\rangle }{p_{+}}+\frac{\langle -\vert \hat \rho_{\lambda_l}\hat\rho_{\lambda_k}\vert -\rangle }{p_{-}} \right]
\end{align}

with $\mathbf{m} = (m_1, \dots, m_N)$, $p_\mathbf{m} = p_{m_1} \cdots p_{m_N}$ and $|\mathbf{m}\rangle = |m_1\rangle \otimes \cdots \otimes |m_N\rangle$, where $m_i \in \{+, -\}$ for $i = 1, \dots, N$. Due to the factorization of the eigenbasis of $\hat\rho_\lambda^{\otimes N}$, the terms in the product in \eqref{eq:quBa} are independent of $i$, which simplifies the expression to a power of $N$, i.e., $\prod_i^N[\,\cdots]=[\,\cdots]^N$. Now we simplify $\hat \rho_{\lambda_l}\hat\rho_{\lambda_k}$
\begin{align*}
\hat \rho_{\lambda_l}\hat\rho_{\lambda_k} &= \frac{1}{4} (\mathbb{I} + \mathbf{r}_l \cdot \boldsymbol{\sigma})(\mathbb{I} + \mathbf{r}_k \cdot \boldsymbol{\sigma}) \\
&= \frac{1}{4} \left[ \mathbb{I} + \mathbf{r}_l \cdot \boldsymbol{\sigma} + \mathbf{r}_k \cdot \boldsymbol{\sigma} + (\mathbf{r}_l \cdot \boldsymbol{\sigma})(\mathbf{r}_k \cdot \boldsymbol{\sigma}) \right] \\
&= \frac{1}{4} \left[ (1 + \mathbf{r}_l \cdot \mathbf{r}_k)\mathbb{I} + (\mathbf{r}_l + \mathbf{r}_k + i \mathbf{r}_l \times \mathbf{r}_k) \cdot \boldsymbol{\sigma} \right]
\end{align*}
where we have used $(\mathbf{a} \cdot \boldsymbol{\sigma})(\mathbf{b} \cdot \boldsymbol{\sigma}) = (\mathbf{a} \cdot \mathbf{b})\mathbb{I} + i(\mathbf{a} \times \mathbf{b}) \cdot \boldsymbol{\sigma}$. Noting that $\langle \pm | A | \pm \rangle = \text{Tr}(|\pm\rangle\langle\pm| A)$ and defining $\mathbf{v}=\frac{\mathbf{r}_\lambda}{|\mathbf{r}_\lambda|} $, the numerator of Eq.~\eqref{eq:quBa} reduces to,
\begin{align*}
\langle \pm | \hat \rho_{\lambda_l}\hat\rho_{\lambda_k} | \pm \rangle &= \text{Tr} \left( \frac{1}{2}(\mathbb{I} \pm \mathbf{v} \cdot \boldsymbol{\sigma}) \cdot \frac{1}{4} \left[ (1 + \mathbf{r}_l \cdot \mathbf{r}_k)\mathbb{I} + (\mathbf{r}_l + \mathbf{r}_k + i \mathbf{r}_l \times \mathbf{r}_k) \cdot \boldsymbol{\sigma} \right] \right) \\
&= \frac{1}{8} \left[ 2(1 + \mathbf{r}_l \cdot \mathbf{r}_k) \pm 2 \mathbf{v} \cdot (\mathbf{r}_l + \mathbf{r}_k + i \mathbf{r}_l \times \mathbf{r}_k) \right] \\
&= \frac{1}{4} \left[ 1 + \mathbf{r}_l \cdot \mathbf{r}_k \pm  \mathbf{v} \cdot (\mathbf{r}_l + \mathbf{r}_k + i \mathbf{r}_l \times \mathbf{r}_k) \right] ,
\end{align*}

where we used $\text{Tr}(\sigma_i) = 0$ and $\text{Tr}((\mathbf{a} \cdot \boldsymbol{\sigma})(\mathbf{b} \cdot \boldsymbol{\sigma})) = 2(\mathbf{a} \cdot \mathbf{b})$. Therefore, the RD QBB matrix elements for $N$ qubits are,
\begin{align}
     (Q_{\mathrm{BB}}^{(R)})_{kl}  =  \frac{1}{2^N}\left[\frac{1 + \mathbf{r}_l \cdot \mathbf{r}_k +  \mathbf{v} \cdot (\mathbf{r}_l + \mathbf{r}_k + i \mathbf{r}_l \times \mathbf{r}_k) }{1 +|\mathbf{r}_\lambda|}+\frac{1 + \mathbf{r}_l \cdot \mathbf{r}_k -  \mathbf{v} \cdot (\mathbf{r}_l + \mathbf{r}_k + i \mathbf{r}_l \times \mathbf{r}_k)}{1 - |\mathbf{r}_\lambda|} \right]^N.
\end{align}
In Fig.~\ref{fig:Qubit}, we compare the SLD QCRB and the RD QBB and SD QBB for $N$-qubit phase estimation. To compute RD QBB and SD QBB the value of $\boldsymbol{\gamma}$ is chosen to impose zero bias at all test points. Although RD QBB is lower than SD QBB for a fixed number of test points,  both bounds exhibit qualitatively similar scaling with the number of copies $N$. While all bounds converge in the asymptotic limit (large $N$), a significant gap remains between the Barankin bounds and the SLD QCRB for small $N$.

\begin{figure}[htbp]
\centering
\includegraphics[width=0.63\textwidth]{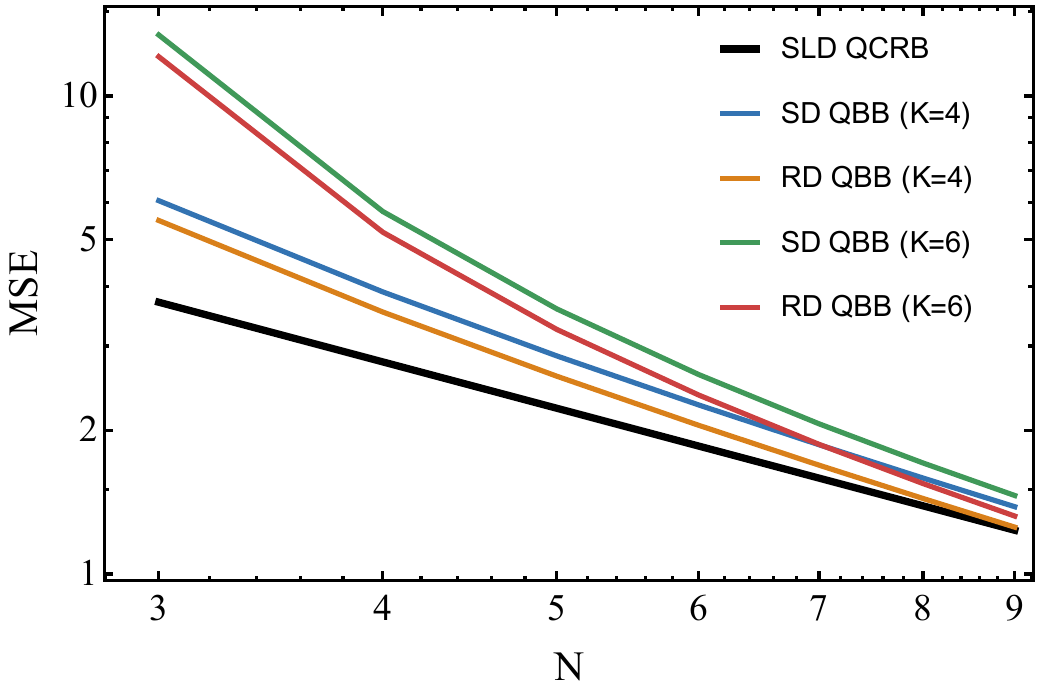}
\caption{Comparison of the RD QBB, SD QBB, and SLD QCRB for qubit phase estimation with parameters $\mathbf{n}=(0,0,1)$, $ \mathbf{r}_0=\frac{1}{10}(0,3,0)$ and $\lambda=\pi/4 $. The $y$-axis shows the MSE while the $x$-axis shows the number of copies $N$. To avoid the potential divergences studied in \cite{gebhart2024fundamental,navarro2025existence}, we only show the region where the bound is finite ($N > 2$). The $K$ test points are equidistributed in the interval $[0,\frac{\pi}{2})$, $\lambda_j=\frac{\pi}{2}\frac{j}{K}$ where $j=0,1,\dots,K-1$. As is common, the unbiasedness condition is imposed at all test points.}
\label{fig:Qubit}
\end{figure}

\subsection{Phase sensing with Gaussian states}
We now consider a single-mode displaced thermal state  $\rho_\lambda $ with first moment $\mathbf{s} _\lambda= \alpha(\cos(\lambda),\sin(\lambda))^T$ and covariance matrix $V_\lambda=\nu I_2$ \cite{fadel2025quantum}. The parameter $\nu \geq 1$ quantifies the thermal noise of the state and is inversely related to its purity. For a system in equilibrium with a thermal bath at temperature $T$, this parameter is given by
\begin{align}
    \nu=\frac{1 + 1/e^{\omega \beta}}{1 - 1/e^{\omega\beta}}=\coth(\frac{\beta\omega}{2}) ,
\end{align}
where $\beta=1/T$ and $\omega$ is the oscillator frequency \cite{serafini2017a}. Since phase sensing consists of the estimation of the phase of the displacement, we assume $\omega\beta=1$ for simplicity. In Fig. \ref{fig:Pha}, we plot the RLD QCRB, SLD QCRB, and the RD QBB when zero bias at all test points is imposed. Consistent with the previous example, the RD QBB and the QCRBs differ significantly in the low-signal regime ($\alpha \ll 1$). This discrepancy is a manifestation of the threshold behavior, which suggests that the MSE predicted by the local SLD CRB is unattainable for globally unbiased estimation when the signal is weak. Remarkably, as the displacement increases ($\alpha \approx 1$), the RD QBB converges to the RLD QCRB and remains below the SLD QCRB.

\begin{figure}[htbp]
\centering
\includegraphics[width=0.65\textwidth]{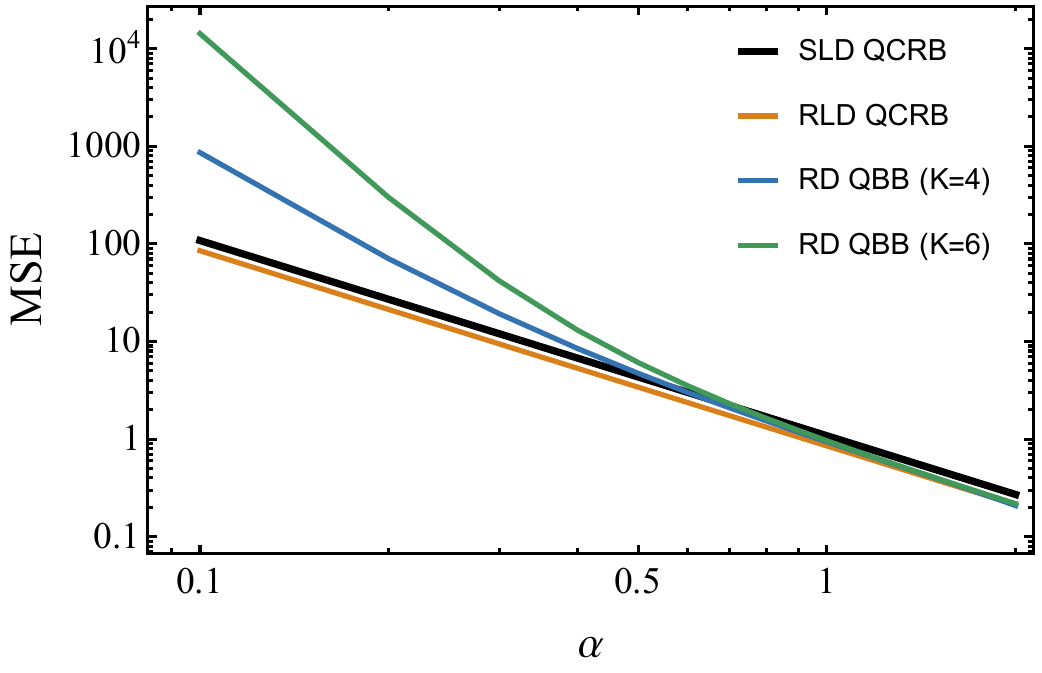}
\caption{Single mode phase estimation with $ \lambda=\frac{\pi}{4}$ and $\omega\beta=1$.  The $y$-axis shows the MSE, while the $x$-axis shows the displacement magnitude $\alpha$. The $K$ test points are equidistributed in the interval   $[0,\frac{\pi}{2})$, $\lambda_j=\frac{\pi}{2}\frac{j}{K}$ where $j=0,1,\dots,K-1$. At all test points $\lambda_j$ the unbiasedness conditions is forced.}. 
\label{fig:Pha}
\end{figure}

\subsection{Temperature sensing}
Thermometry involves the estimation of the temperature $T$ of a thermal state \cite{correa2015individual}. However, it is standard practice to estimate the inverse temperature $\beta = 1/T$ instead \cite{mehboudi2019thermometry,cenni2022thermometry}. For consistency with our notation, we identify the parameter of interest as $\lambda = \beta$. Consequently, the first and second moments of a single-mode state $\hat{\rho}_\lambda$ are given by $\mathbf{s}_\lambda = (0,0)^T$ and $V_\lambda = \coth(\lambda/2) \mathbb{I}_2$, respectively, where we have set $\omega = 1$ for simplicity.
In Fig.~\ref{fig:Tb}, we plot the RD QBB (assuming zero bias at all test points) and the SLD QCRB for the multicopy states $\hat{\rho}_\lambda^{\otimes N} $. As in the previous examples, we observe a clear threshold behavior. Nevertheless, for $N > 50$, the $K=2$ RD QBB is significantly smaller than the SLD QCRB. As stated before, in the high-signal regime ($N \gg 1$), the SLD QCRB is an informative bound because it can always be saturated. For large $N$, the scaling agreement between the RD QBB and the SLD QCRB improves as $K$ increases. Additionally, Fig.~\ref{fig:Ta} illustrates the comparison between the RLD QCRB and the RD QBB using two test points, $(\lambda_1, \lambda_2) = (2, 2 + \Delta \lambda)$. As discussed in Sec.~\ref{subsec:RecoverRLDCRBgeneral}, in the limit $\Delta \lambda \to 0$, the RD QBB converges to the RLD QCRB when the unbiasedness condition is imposed at the test points.

\begin{figure}[H] 
    \centering
    \subfloat[Comparison between the RLD QCRB and the RD QBB with two test points $(\lambda_1, \lambda_2) = (2, 2 + \Delta \lambda)$ for a single copy ($N = 1$). In the limit $\Delta \lambda \to 0$, both bounds converge to the same value.]{%
        \includegraphics[width=0.475\textwidth]{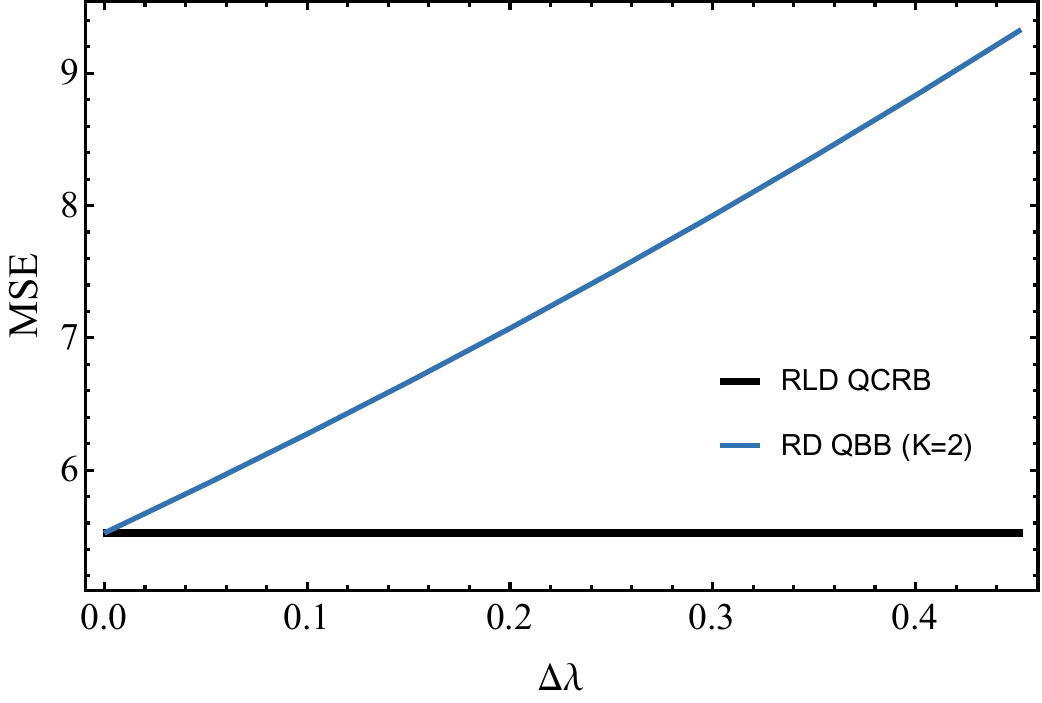}%
        \label{fig:Ta}%
        }%
    \hfill%
    \subfloat[ Comparison between the SLD QCRB and the RD QBB with different numbers of test points $K$, chosen as $\lambda_i = 2 + \frac{i-1}{K}$ for $i = 1, \dots, K$.  ]{%
        \includegraphics[width=0.495\textwidth]{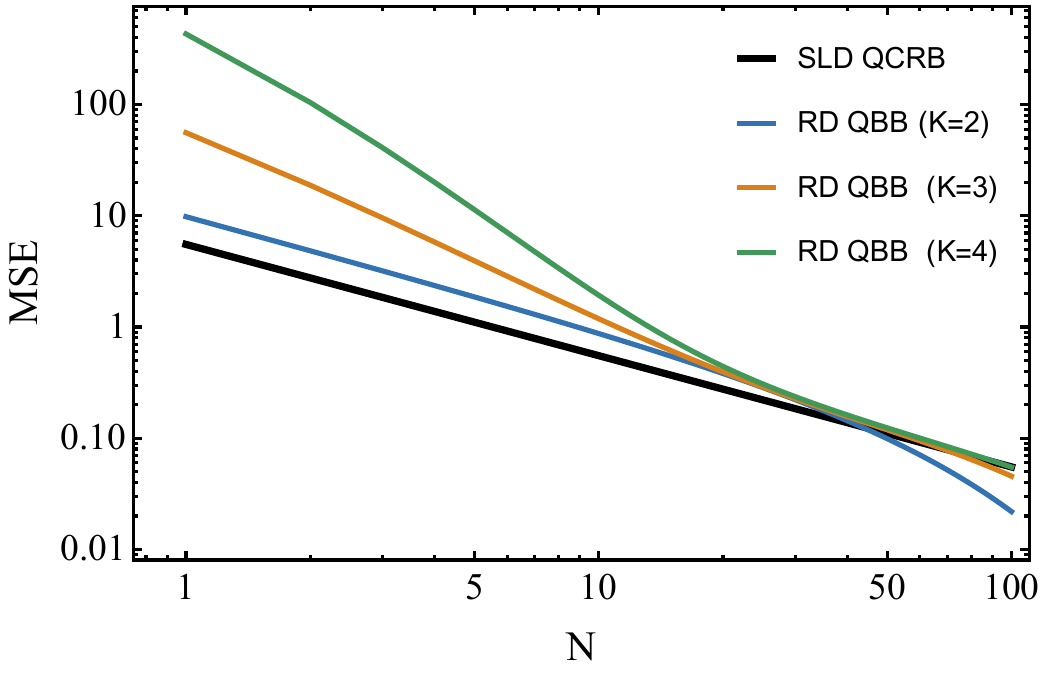}%
        \label{fig:Tb}%
        }%
    \caption{Quantum thermometry with $N$ copies, assuming $\omega = 1$ and a true parameter value $\lambda = 2$. For all RD QBB the value of $\boldsymbol{\gamma}$ is chosen to force the  unbiasedness condition at all test points.}\label{fig:Thermo}
\end{figure}

\subsection{Loss sensing}
We now consider an example where the parameter of interest is encoded in both the displacement and the covariance matrix. An important model in quantum optics is the loss channel, which describes the attenuation of a signal due to its interaction with a thermal environment. The precision limits for this model, characterized by the SLD QCRB, have been extensively studied for both Gaussian probe states \cite{monras2007optimal,navarro2025super,PhysRevA.109.053715} and arbitrary input states \cite{adesso2009optimal}. Here, we analyze this estimation problem using the RD QBB to investigate the conditions under which the SLD QCRB might be unattainable. The action of the loss channel on the first and second moments is defined as follows,
\begin{align}
    \mathbf{s} \to \sqrt{\lambda}\:  \mathbf{s} \eqqcolon  \mathbf{s}_{\lambda}\hspace{1cm} V \to \lambda V + (1-\lambda)\mu I \eqqcolon V_{\lambda}. 
\end{align}
Here, $\mu I$ represents the environment's thermal state. At zero temperature $\mu = 1$, while for finite temperatures $\mu > 1$. The losses are characterized by the transmissivity parameter $\lambda \in [0, 1]$, where $\lambda=1$ corresponds to a noiseless (identity) channel and $\lambda=0$ corresponds to a completely lossy channel. Our objective is to estimate $\lambda$ given an initial probe state.\\
We consider a single-mode coherent state as the input, with $\mathbf{s}=(\alpha,0)^T$, $V=I_2$, and an environment characterized by $\mu = \coth(1/2)$. In Fig.~\ref{fig:Db}, we verify the expected agreement between the RD QBB and the RLD QCRB as $\Delta \lambda \to 0$ when the unbiasedness condition is imposed at the test points. For loss sensing with the multicopy state $\hat{\rho}^{\otimes N}$, as shown in Fig.~\ref{fig:Da}, the threshold gap is less significant than in previous examples. Notably, as the number of test points $K$ is increased, we observe the convergence of the RD QBB. As discussed previously, this convergence reflects the attainment of the supremum over the set of test point, characterizing the performance limit for estimators that are globally unbiased across the parameter range. In all calculations, the vector $\boldsymbol{\gamma}$ was fixed to ensure the zero-bias condition at every test point.
\begin{figure}[H] 
    \centering
        \subfloat[Comparison between the RLD QCRB and the RD QBB with two test points $(\lambda_1, \lambda_2) = (1/4, 1/4 + \Delta \lambda)$ for a single copy ($N=1$). In the limit $\Delta \lambda \to 0$, both bounds converge to the same value.]{%
        \includegraphics[width=0.47\textwidth]{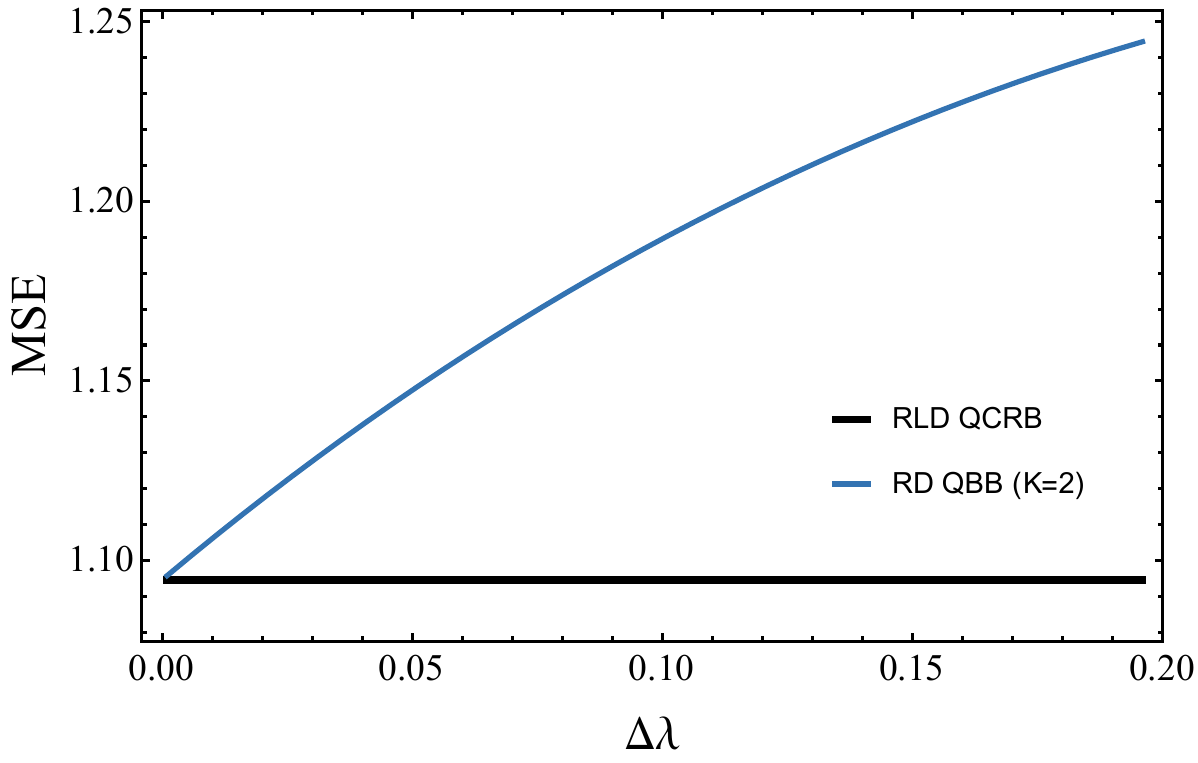}%
        \label{fig:Db}%
        }%
    \hfill%
    \subfloat[Comparison between the SLD QCRB, RLD QCRB, and the RD QBB with $K$ test points given by $\lambda_i = (i-1)/K$ for $i=1, \dots, K$. The overlap between RD QBB $K=10$, $K=20$ and $K=30$ illustrates the approach toward the limit for globally unbiased estimators.]{%
        \includegraphics[width=0.465\textwidth]{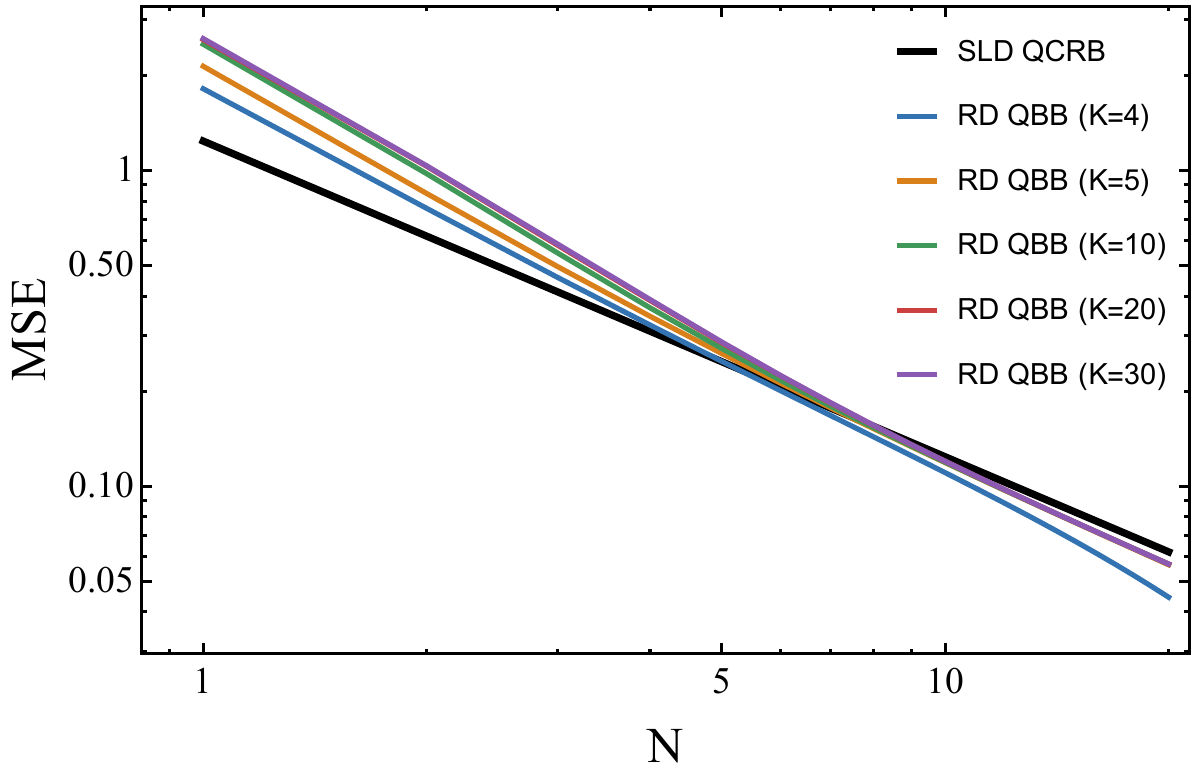}%
        \label{fig:Da}%
        }%
    \caption{Loss sensing with $N$ copies, assuming a true transmissivity $\lambda = 1/4$ and displacement amplitude $\alpha = 1/2$. For all RD QBB the value of $\boldsymbol{\gamma}$ is chosen to force the  unbiasedness condition at all test points.}\label{fig:Dissipative}
\end{figure}

\section{Conclusions}
We have derived a new quantum version of the Barankin bound based on a right-acting division operator. We derived sufficient existence conditions, considered the attainability by establishing a hierarchy between SD QBB and the RD QBB, and demonstrated that the RLD QCRB can be recovered from the RD QBB.

We applied our general results to the class of bosonic Gaussian states and reformulated the sufficient existence conditions and the Barankin matrix elements in terms of the mean vectors and covariance matrices, which allows for an efficient evaluation of the bound. Future work could be directed towards finding such analytic formulas for the Barankin matrix elements of the SD QBB in Ref.~\cite{gessner2023hierarchies}.

We then calculated the RD QBB for various examples, such as phase, temperature, and loss sensing with Gaussian states and compared them to known bounds. In all of our examples, we observed the threshold effect, i.e., a departure of the quantum Barankin bound from the quantum Cram\'er-Rao bound in the few-shot/low-energy regime. This shows that the quantum Cram\'er-Rao bound is  too optimistic in these regimes for globally unbiased estimation, and our RD QBB provides a tighter and more informative bound.

As perspectives, we note that the Barankin bound historically has been used to study estimation in the radar problem~\cite{mcaulay1971a}. Therefore, it would be interesting to apply the RD QBB to quantum radar \cite{giovannetti2001quantum,tan2008quantum,zhuang2017entanglement,reichert2022quantum,reichert2024heisenberg,reichert2026resource} to study the threshold behavior. First steps in the frequentist framework have been taken in Ref.~\cite{Reichert2026}, and in the Bayesian framework Ref.~\cite{zhuang2022a} studied the problem by employing the Ziv-Zakai bound. Ultimately, thanks to the generality of our RD QBB approach, all bosonic Gaussian state estimation protocols can now be reexamined in the few-shot regime, to determine the threshold behavior and determine the ultimate limits for unbiased estimation.

\subsection*{Acknowledgements}
M.~R. thanks Francesco Albarelli and Ludovico Lami for helpful discussions and acknowledges support from UPV/EHU Ph.D. Grant No. PIF21/289. G.~A. acknowledges support from EPSRC Grant No. EP/X010929/1.
This project was supported by the Basque Government BasQ initiative under the Q-STREAM project. We also acknowledge support from OpenSuperQ+100 (Grant No. \allowbreak{101113946}) of the EU Flagship on Quantum Technologies, from Project Grant No. PID2024-156808NB-I00 and Spanish Ram\'on y Cajal Grant No. RYC-2020-030503-I funded by MI-CIU/AEI/10.13039/501100011033 and by “ERDF A way of making Europe” and “ERDF Invest in your Future”, and from the Spanish Ministry for Digital Transformation and of Civil Service of the Spanish Government through the QUANTUM ENIA project call Quantum Spain, and by the EU through the Recovery, Transformation and Resilience Plan–Next Generation EU within the framework of the Digital Spain 2026 Agenda.  
This work was also funded by the project PID2023-152724NA-I00, with funding from MCIU/AEI/10.13039/501100011033 and FSE+, by the project CNS2024-154818 with funding by MICIU/AEI /10.13039/501100011033, by the project CIPROM/2022/66 with funding by the Generalitat Valenciana, and by the Ministry of Economic Affairs and Digital Transformation of the Spanish Government through the QUANTUM ENIA Project call—QUANTUM SPAIN Project, by the European Union through the Recovery, Transformation and Resilience Plan—NextGenerationEU within the framework of the Digital Spain 2026 Agenda, and by the CSIC Interdisciplinary Thematic Platform (PTI+) on Quantum Technologies (PTI-QTEP+). This work is supported through the project CEX2023-001292-S funded by MCIU/AEI.

\appendix

\section{Useful relations for quadratic exponential operators}\label{app:GaussianRelations}
Here, we list useful properties that were stated or derived in Ref.~\cite{seshadreesan2018renyi} so that we can refer to them in our proofs.

\subsection{Property}
 The displacement operator has the following properties:
\begin{align}
    \hat D(\mathbf s)^{-1} &= \hat D(-\mathbf s)  \\
    \hat D (\mathbf s)\hat D(\mathbf t) &= \hat D(\mathbf s+\mathbf t) e^{-\frac{1}{2}\mathbf s^T i\Omega \mathbf t} \label{eq:DispProd} \\
    \hat D(\mathbf s)\hat{ \mathbf x}\hat D(-\mathbf s) &= \hat {\mathbf x} + \mathbf s  .\label{eq:quadTrafDisp}
\end{align}
Note that from the above follows $  \hat D (\mathbf t)\hat D(\mathbf s) = \hat D(\mathbf s+\mathbf t) e^{-\frac{1}{2}\mathbf t^T i\Omega \mathbf s} $ from which follows $\hat D(\mathbf s+\mathbf t)  = \hat D (\mathbf t)\hat D(\mathbf s)e^{+\frac{1}{2}\mathbf t^T i\Omega \mathbf s}$ .

\subsection{Property}
Corollary 14 in \cite{seshadreesan2018renyi} states that
\begin{align}
    \hat D (\mathbf l) \exp[-\frac{1}{2}\hat{\mathbf  x}^T H \hat{ \mathbf x}] = \hat D(-\mathbf s) \exp[-\frac{1}{2}\hat{ \mathbf x}^T H \hat{\mathbf  x}] \hat D (\mathbf s) e^{\frac{1}{4}\mathbf l^T i\Omega W \mathbf l}, \label{eq:DispSand}
\end{align}
where $H$ is a complex symmetric matrix and $\mathbf l= (\exp[-i\Omega H]-I)\mathbf s$,
where $\mathbf s\in\mathds{C}^{2n}$, and 
\begin{align}
    W = \frac{I+\exp(i\Omega H)}{I-\exp(i\Omega H)}
\end{align}
and we have
\begin{align}
    W= - V i\Omega .
\end{align}
Also note that it follows that $i\Omega W = i\Omega (-V i\Omega) = + \Omega V\Omega= -\Omega^T V \Omega$.

\subsection{Property}
Lemma 16 of \cite{seshadreesan2018renyi} stated that
for positive-definite real matrices $H_1,H_2$ such that $e^{-i\Omega H_1}e^{-i\Omega H_2} = e^{-i\Omega H_3}$ (which corresponds to $e^{\hat H_1} e^{\hat H_2} = e^{\hat H_3}$, see proposition 6 of \cite{seshadreesan2018renyi}) and $\mathbf l = (\exp(-i\Omega H_1)-I) \mathbf{s}$, the following holds
\begin{align}
      -\frac{1}{4}\mathbf l^T i\Omega W_1 \mathbf l +\frac{1}{4}\mathbf  l^T i\Omega   W_3 \mathbf  l  =  -\mathbf  s^T (V_1+V_2)^{-1} \mathbf  s  .\label{eq:PropertyW}
\end{align}
 Note that if we replace $H_2\rightarrow -H_2$, we simply replace everywhere $V_2\rightarrow -V_2$ and $W_2\rightarrow -W_2$ as argued in Lemma 7 of Ref.~\cite{seshadreesan2018renyi}. In this case, we get $   -\frac{1}{4}\mathbf l^T i\Omega W_1 \mathbf l +\frac{1}{4}\mathbf  l^T i\Omega   W_3 \mathbf  l  =  -\mathbf  s^T (V_1-V_2)^{-1} \mathbf  s $ if $V_1-V_2$ invertible.

\section{Properties of products of Gaussian exponential operators}\label{app:PropertiesProduc}
\begin{theorem}\label{theorem:TraceClassCondition}
    Let $\hat \rho_1 = e^{\hat H_1}$ and $\hat \rho_2 = e^{\hat H_2}$ be exponential operators (unnormalized Gaussian states) with $\hat H_i = - \frac{1}{2} \hat{\mathbf x}^T H_i \hat{\mathbf{x}}$ where $H_i> 0$. They have covariance matrices $V_{\hat \rho_1}$ and $V_{\hat\rho_2}$. The operator
    \begin{align}
        \hat \rho_1 \hat \rho_2^{-1} = e^{\hat H_1} e^{-\hat H_2}
    \end{align}
    is trace class if \begin{align}
      2( V_{\hat\rho_2^2}  - V_{\hat\rho_1^2} ) =   V_{\hat\rho_2} - V_{\hat\rho_1} +\Omega^T (V_{\hat\rho_2}^{-1} - V_{\hat\rho_1}^{-1}) \Omega > 0 .
    \end{align}
\end{theorem}
\begin{proof}
    Define $\hat A = e^{\hat H_1}e^{-\hat H_2}$. Now, $\hat A$ is trace class iff $\vert A\vert = \sqrt{\hat A^\dag \hat A}$, or equivalently $\sqrt{\hat A\hat A^\dag}$, is trace class. So, let us consider
\begin{align}
     \hat A \hat A^\dag&= e^{\hat H_1}e^{-\hat H_2} (e^{\hat H_1}e^{-\hat H_2})^\dag  \\
    &=   e^{\hat H_1}e^{-\hat H_2}e^{-\hat H_2} e^{\hat H_1}  \\
    &=e^{\hat H_1} e^{-2\hat H_2} e^{\hat H_1} .
\end{align}
Now, we use Proposition 12 of Ref.~\cite{seshadreesan2018renyi}, which allows us to combine the three exponential operators into a single one
\begin{align}
    \hat A \hat A^\dag&= e^{\hat H_1} e^{-2\hat H_2} e^{\hat H_1} =  e^{\hat H'} ,
\end{align}
where $\hat H' = -\frac{1}{2} \hat x^T H' \hat x$ and $H'>0$ is defined as $H' =  2i\Omega \mathrm{arccoth}( V' i\Omega)$ with
\begin{align}
    V' = V_{\hat\rho_1^2} + \left( I+(V_{\hat\rho_1^2}\Omega)^{-2}\right)V_{\hat\rho_1^2} \left(V_{\hat\rho_2^2}- V_{\hat\rho_1^2}\right)^{-1} V_{\hat\rho_1^2} \left( I+(\Omega V_{\hat\rho_1^2})^{-2}\right) ,
\end{align}
where
\begin{align}
    V_{\hat\rho_i^2} = \frac{1}{2} \left( V_{\hat\rho_i} - \Omega V_{\hat\rho_i}^{-1} \Omega \right) =  \frac{1}{2} \left( V_{\hat\rho_i} + \Omega^T V_{\hat\rho_i}^{-1} \Omega \right)
\end{align}
is the covariance matrix of the squared state $(e^{\hat H_i})^2 = e^{2\hat H_i}$. Recall the relations in Eqs.~\eqref{eq:HofV} and \eqref{eq:VofH}  between  $H_i $ and $V_i$.

According to Proposition 12 of Ref.~\cite{seshadreesan2018renyi}, the operator $\hat A\hat A^\dag = e^{\hat H'}$ is an unnormalized Gaussian state, and thus trace class, if $V_{\hat\rho_2^2}> V_{\hat\rho_1^2}$. Rearranging yields
\begin{align}
  2\left(  V_{\hat\rho_2^2}  - V_{\hat\rho_1^2}  \right) &=  V_{\hat\rho_2} + \Omega^T V_{\hat\rho_2}^{-1} \Omega - V_{\hat\rho_1} - \Omega^T V_{\hat\rho_1}^{-1} \Omega \\
  &= V_{\hat\rho_2}-V_{\hat\rho_1} + \Omega^T (V_{\hat\rho_2}^{-1}-V_{\hat\rho_1}^{-1}) \Omega \\
  & > 0 ,
\end{align}
where we used Eq.~\eqref{eq:SquareCov}.
Thus $\hat A\hat A^\dag =  e^{\hat H'}$ is a (unnormalized) Gaussian state if $V_{\hat\rho_2}-V_{\hat\rho_1} + \Omega^T (V_{\hat\rho_2}^{-1}-V_{\hat\rho_1}^{-1}) \Omega> 0$. Thus, if  $V_{\hat\rho_2}-V_{\hat\rho_1} + \Omega^T (V_{\hat\rho_2}^{-1}-V_{\hat\rho_1}^{-1}) \Omega > 0$, the square root of this operator $ \sqrt{\hat A\hat A^\dag } = e^{\frac{1}{2}\hat H'}$ is also an unnormalized Gaussian state   as $\frac{1}{2}H' >0$. As Gaussian states are trace class, we also have that $\sqrt{A A^\dag}$ is trace class, from which follows that $A = e^{\hat H_1}e^{-\hat H_2}$ is trace class if  $V_{\hat\rho_2}-V_{\hat\rho_1} + \Omega^T (V_{\hat\rho_2}^{-1}-V_{\hat\rho_1}^{-1}) \Omega > 0$.
\end{proof}

\begin{theorem}\label{theorem:TraceValue}
      Let $\hat \rho_1 = e^{\hat H_1}$ and $\hat \rho_2 = e^{\hat H_2}$ be exponential operators (unnormalized Gaussian states) with $\hat H_i = - \frac{1}{2} \hat{\mathbf x}^T H_i \hat{\mathbf{x}}$ where $H_i> 0$. They have covariance matrices $V_{\hat \rho_1}$ and $V_{\hat\rho_2}$ and let \begin{align}
        V_{\hat\rho_2} - V_{\hat\rho_1} +\Omega^T (V_{\hat\rho_2}^{-1} - V_{\hat\rho_1}^{-1}) \Omega > 0 .
    \end{align}
    We then have
    \begin{align}
        \mathrm{Tr}[e^{ \hat H_1} e^{-\hat H_2} ] = \sqrt{\frac{\mathrm{Det}[\frac{ V_{\hat\rho_1}+i\Omega}{2}]\mathrm{Det}[\frac{ V_{\hat\rho_2}+i\Omega}{2}]}{ \mathrm{Det}[\frac{ V_{\hat\rho_2}- V_{\hat\rho_1}}{2}]}} .
    \end{align}
\end{theorem}
\begin{proof}
Because of the assumption $ V_{\hat\rho_2} - V_{\hat\rho_1} +\Omega^T (V_{\hat\rho_2}^{-1} - V_{\hat\rho_1}^{-1}) \Omega > 0$, it follows from Theorem~\ref{theorem:TraceClassCondition} that $e^{\hat H_1}e^{-\hat H_2}$ is trace class.
     For a trace-class operator, the value of the trace is independent of the chosen orthonormal basis. Let us choose the convenient orthonormal basis $\vert \mathbf{n}\rangle$ which is the eigenbasis of $e^{\hat H_1}$, i.e., $e^{\hat H_1}\vert \mathbf{n}\rangle = c_{\mathbf{n}} \vert \mathbf{n}\rangle$. With this, we find:
\begin{align}
    \text{Tr}[  e^{\hat H_1}e^{-\hat H_2}]  = \sum_{\mathbf{n}} \langle \mathbf{n}\vert e^{\hat H_1}e^{-\hat H_2}\vert \mathbf{n}\rangle  
    &= \sum_{\mathbf{n}} \langle \mathbf{n}\vert e^{\frac{1}{2}\hat H_1}\left( e^{\frac{1}{2}\hat H_1} e^{-\hat H_2} e^{\frac{1}{2}\hat H_1}\right) e^{-\frac{1}{2}\hat H_1}\vert \mathbf{n}\rangle  \\
    &= \sum_{\mathbf{n}} \langle \mathbf{n}\vert c_{\mathbf{n}}^{\frac{1}{2}}\left( e^{\frac{1}{2}\hat H_1} e^{-\hat H_2} e^{\frac{1}{2}\hat H_1}\right) c_{\mathbf{n}}^{-\frac{1}{2}}\vert \mathbf{n}\rangle  \\
     &= \sum_{\mathbf{n}} \langle \mathbf{n}\vert \left( e^{\frac{1}{2}\hat H_1} e^{-\hat H_2} e^{\frac{1}{2}\hat H_1}\right)  \vert \mathbf{n}\rangle  \\
     &= \text{Tr}[e^{\frac{1}{2}\hat H_1} e^{-\hat H_2} e^{\frac{1}{2}\hat H_1}] = \text{Tr}[e^{\hat H_3}] .
\end{align}
In the first line after the second equals sign, we used $e^{\hat H_1} =e^{\frac{1}{2}\hat H_1}e^{\frac{1}{2}\hat H_1}$ and furthermore we inserted an identity  $\vert \mathbf{ n}\rangle = e^{\frac{1}{2}\hat H_1}e^{-\frac{1}{2}\hat H_1}\vert \mathbf{ n}\rangle$ . Note that inserting an identity $BB^{-1}$ can be problematic when $B^{-1}$ is unbounded. In our case, however, the operator $B^{-1}$ acts on a state that is in its domain and therefore there are no issues. In the second line, we used that $e^{\pm \frac{1}{2}\hat H_1}\vert\mathbf{n}\rangle = c_{\mathbf{n}}^{\pm \frac{1}{2}} \vert\mathbf{n}\rangle$, following from $e^{\hat H_1}\vert \mathbf{n}\rangle = c_{\mathbf{n}} \vert \mathbf{n}\rangle$. In the third line we canceled the coefficients.  In the fourth line, we introduced $\hat H_3 = -\frac{1}{2} \mathbf{\hat x}^T H_3\mathbf{\hat x}$, where $H_3$ is given by $e^{i\Omega H_3} =e^{i\Omega H_2/2} e^{i\Omega (-H_1)}e^{i\Omega H_2/2}$ \cite{seshadreesan2018renyi}. Now, since
$
\left|\operatorname{Tr}\!\left[e^{\hat H_1}e^{-\hat H_2}\right]\right|
=
\left|\operatorname{Tr}\!\left[e^{\frac12\hat H_1}e^{-\hat H_2}e^{\frac12\hat H_1}\right]\right|
<\infty,
$
it follows that \(e^{\hat H_3}\) is trace class, and hence \(H_3>0\) (note that  $e^{\hat H_3} \text{ is trace class}\Leftrightarrow H_3>0$). Now, Ref.~\cite{seshadreesan2018renyi}  calculated $\text{Tr}[e^{\hat H_3}]=  \sqrt{\frac{\mathrm{Det}[\frac{ V_{\hat\rho_1}+i\Omega}{2}]\mathrm{Det}[\frac{ V_{\hat\rho_2}+i\Omega}{2}]}{ \mathrm{Det}[\frac{ V_{\hat\rho_2}- V_{\hat\rho_1}}{2}]}}$ in Proposition 12, by assuming $V_2>V_1$ without assuming $H_3>0$. However, their proof still works by using the assumption $H_3>0$ instead of $V_2>V_1$. This can be seen by noting that $H_3>0$ implies that $I-e^{i\Omega H_3} = I-e^{i\Omega H_2/2} e^{i\Omega (-H_1)}e^{i\Omega H_2/2} $, which appears in Eq.~(121) of Ref.~\cite{seshadreesan2018renyi}, is invertible, which then allows us to use the same steps as in Ref.~\cite{seshadreesan2018renyi} to derive $\text{Tr}[e^{\hat H_3}]=  \sqrt{\frac{\mathrm{Det}[\frac{ V_{\hat\rho_1}+i\Omega}{2}]\mathrm{Det}[\frac{ V_{\hat\rho_2}+i\Omega}{2}]}{ \mathrm{Det}[\frac{ V_{\hat\rho_2}- V_{\hat\rho_1}}{2}]}}$ only using $H_3>0$.
% Also note that $I\pm e^{i\Omega G}$ is invertible if $G$ real symmetric with $G>0$ or $G$, which is somethin used in Eq.~(89) and Eq.~(90).
In fact, it also follows that $V_2>V_1$ whenever $H_3 >0$, equivalently whenever $e^{\hat H_3}$ is trace class. We thus have
\begin{align}
      \text{Tr}[  e^{\hat H_1}e^{-\hat H_2}]  = \sqrt{\frac{\mathrm{Det}[\frac{ V_{\hat\rho_1}+i\Omega}{2}]\mathrm{Det}[\frac{ V_{\hat\rho_2}+i\Omega}{2}]}{ \mathrm{Det}[\frac{ V_{\hat\rho_2}- V_{\hat\rho_1}}{2}]}} .
\end{align}
\end{proof}
\begin{corollary}\label{corollary:Invertibility}
    Let $V_{\hat\rho_1}$ and $V_{\hat\rho_2}$ be covariance matrices of (unnormalized) Gaussian states $\hat \rho_1 = e^{\hat H_1}$ and $\hat \rho_2 = e^{\hat H_2}$ that satisfy $ V_{\hat\rho_2} - V_{\hat\rho_1} +\Omega^T (V_{\hat\rho_2}^{-1} - V_{\hat\rho_1}^{-1}) \Omega > 0$ or equivalently $V_{\hat\rho_2^2} - V_{\hat\rho_1^2} >0$. It then follows that 
    \begin{align}
        V_{\hat\rho_2} -V_{\hat\rho_1}
    \end{align}
    is invertible.
\end{corollary}
\begin{proof}
This is an immediate consequence of Theorem~\ref{theorem:TraceClassCondition} and \ref{theorem:TraceValue}, where we showed that $V_{\hat\rho_2} - V_{\hat\rho_1} +\Omega^T (V_{\hat\rho_2}^{-1} - V_{\hat\rho_1}^{-1}) \Omega > 0$ implies that $e^{ \hat H_1} e^{-\hat H_2}$ is trace class, from which eventually follows $ V_{\hat\rho_2} -V_{\hat\rho_1}>0$. Thus, $ V_{\hat\rho_2} -V_{\hat\rho_1}$ is invertible.
\end{proof}

\begin{theorem}\label{theorem:CharacteristicProd}
   Consider the trace-class  operator $\hat O = \left( \hat D(-\mathbf s_1) e^{\hat H_1} \hat D(\mathbf s_1) \right)\left(  \hat D(-\mathbf s_2) e^{\hat H_2} \hat D(\mathbf s_2) \right)$, which is the product of two unnormalized Gaussian states with mean vectors $\mathbf s_1,\mathbf s_2$ and covariance matrices $V_1,V_2$. Note that  $\hat H_i = - \frac{1}{2} \hat{\mathbf x}^T H_i \hat{\mathbf{x}}$, where $H_i> 0$ for $i=1,2$. The characteristic function of $\hat O$ is given by
   \begin{align}
    \chi_{\hat O} (\mathbf z) = e^{- \delta\mathbf s^T (V_1+V_2)^{-1} \delta\mathbf s} e^{- \frac{1}{4}\mathbf z^T \Omega^T V_3 \Omega \mathbf z } e^{i\mathbf{z}^T \Omega^T\left( + \mathbf s_2   + (V_2+i\Omega) (V_2+V_1)^{-1}   \delta \mathbf s  \right)} \sqrt{\frac{\mathrm{Det}[\frac{V_1+i\Omega}{2}]\mathrm{Det}[\frac{V_2+i\Omega}{2}]}{ \mathrm{Det}[\frac{V_2+V_1}{2}]}} ,
\end{align}
where
\begin{align}
    V_3 = -i\Omega  + (V_2+i\Omega) (V_1+V_2)^{-1}(V_1+i\Omega) .
\end{align}
\end{theorem}
\begin{proof}
Recall that the characteristic function is defined as $\text{Tr}\left[  \hat O \hat D(-\mathbf z)\right] $. We thus get
\begin{align}
     \chi_{\hat O}(\mathbf z) &= \text{Tr}\left[ \hat D(-\mathbf s_1) e^{\hat H_1} \hat D(\mathbf s_1)   \hat D(-\mathbf s_2) e^{\hat H_2} \hat D(\mathbf s_2)  \hat D(-\mathbf z)\right] \label{eq:C1}\\
     &= \text{Tr}\left[\hat D(-\mathbf s_2)\hat D(\mathbf s_2)\hat D(-\mathbf s_1) e^{\hat H_1} \hat D(\mathbf s_1)   \hat D(-\mathbf s_2) e^{\hat H_2} \hat D(\mathbf s_2)  \hat D(-\mathbf z)\right] \\
     &= \text{Tr}\left[\hat D(-\mathbf s_2) \hat D(-\delta\mathbf{s}) e^{\hat H_1} \hat D(\delta\mathbf{s})    e^{\hat H_2} \hat D(\mathbf s_2)  \hat D(-\mathbf z)\right] \\
     &= e^{-\frac{1}{4}\mathbf l^T i\Omega W_1 \mathbf l}\text{Tr}\left[\hat D(-\mathbf s_2) \hat D( \mathbf{l}) e^{\hat H_1}      e^{\hat H_2} \hat D(\mathbf s_2)  \hat D(-\mathbf z)\right] \label{eq:C4}\\
      &= e^{-\frac{1}{4}\mathbf l^T i\Omega W_1 \mathbf l}\text{Tr}\left[ \hat D(\mathbf s_2)\hat D(-\mathbf z)\hat D(-\mathbf s_2) \hat D( \mathbf{l})    e^{\hat H_3}   \right] \\
       &= e^{-\frac{1}{4}\mathbf l^T i\Omega W_1 \mathbf l}e^{(-\mathbf z)^T i\Omega \mathbf s_2}\text{Tr}\left[  \hat D(-\mathbf z)  \hat D( \mathbf{l})  e^{\hat H_3}   \right] \\
        &= e^{-\frac{1}{4}\mathbf l^T i\Omega W_1 \mathbf l}e^{(-\mathbf z)^T i\Omega \mathbf s_2}  e^{-\frac{1}{2} (-\mathbf z)^T i\Omega \mathbf l}\text{Tr}\left[     \hat D(-\mathbf z +\mathbf{l})  e^{\hat H_3}    \right] \\
         &= e^{-\frac{1}{4}\mathbf l^T i\Omega W_1 \mathbf l}e^{(-\mathbf z)^T i\Omega \mathbf s_2}  e^{-\frac{1}{2} (-\mathbf z)^T i\Omega \mathbf l} e^{\frac{1}{4}(-\mathbf z +\mathbf{l})^T i\Omega W_3 (-\mathbf z +\mathbf{l})}  \text{Tr}\left[     \hat D(-\mathbf t) e^{\hat H_3} \hat D(\mathbf t)  \right] \\
           &= e^{-\frac{1}{4}\mathbf l^T i\Omega W_1 \mathbf l}e^{(-\mathbf z)^T i\Omega \mathbf s_2}  e^{-\frac{1}{2} (-\mathbf z)^T i\Omega \mathbf l} e^{\frac{1}{4}(-\mathbf z +\mathbf{l})^T i\Omega W_3 (-\mathbf z +\mathbf{l})}  \text{Tr}\left[    e^{\hat H_3}  \right] .
\end{align}
  In the second line, we introduced an identity $\hat I= \hat D(-\mathbf s_2)\hat D(\mathbf s_2)$.
    In the third line, we denote $\delta\mathbf{s} = \mathbf s_1- \mathbf s_2$. We used Eq.~\eqref{eq:DispProd} twice, both exponential functions cancel. We have according to Eq.~\eqref{eq:DispProd} $\hat D(\mathbf s)\hat D(\mathbf t)= \hat D(\mathbf s+\mathbf t)e^{-\frac{1}{2} \mathbf s^T i\Omega \mathbf t}$ and thus also $\hat D(\mathbf t)\hat D(\mathbf s)= \hat D(\mathbf s+\mathbf t)e^{-\frac{1}{2} \mathbf t^T i\Omega \mathbf s}$.  Now consider the exponents. We have $\mathbf s^T \Omega \mathbf t = \mathbf s^T (\Omega \mathbf t) = (\Omega \mathbf t)^T \mathbf s = \mathbf t^T \Omega^T \mathbf s = - \mathbf t^T \Omega \mathbf s$. Thus, the exponents indeed cancel.
    In the fourth line, we used Eq.~\eqref{eq:DispSand} with $\mathbf l= (\exp[-i\Omega H_1]-I)\delta \mathbf s$ and  $W_1 = \frac{I+\exp(i\Omega H_1)}{I-\exp(i\Omega H_1)}$ and  recall $W_1= - V_1 i\Omega$.
    In the fifth line, we used the cyclic property of the trace which can be used because $ e^{\hat H_1}    e^{\hat H_2} = e^{\hat H_3}$ is trace class and the displacement operators are bounded operators.
     In the sixth line, we used $\hat D(\mathbf s_2)\hat D(-\mathbf z)\hat D(-\mathbf s_2) = \hat D(\mathbf s_2)\exp [ (-\mathbf{z})^T i \Omega \hat{ \mathbf{x}}]\hat D(-\mathbf s_2) = \exp [ (-\mathbf{z})^T i \Omega \hat D(\mathbf s_2)\hat{ \mathbf{x}}D(-\mathbf s_2)] =\exp [ (-\mathbf{z})^T i \Omega (\mathbf{x} + \mathbf s_2)]= \hat D(-\mathbf{z}) e^{(-\mathbf z)^T i\Omega \mathbf s_2} $, where we used the definition of the displacement operator in Eq.~\eqref{eq:defDisp} and the relation in Eq.~\eqref{eq:quadTrafDisp}. 
    In the seventh line, we use Eq.~\eqref{eq:DispProd}: $ \hat D(-\mathbf z)  \hat D( \mathbf{l}) =  \hat D(-\mathbf z +\mathbf{l}) e^{-\frac{1}{2} (-\mathbf z)^T i\Omega \mathbf l}$.
    In the eighth line, we again use Eq.~\eqref{eq:DispSand}: $\hat D(-\mathbf z +\mathbf{l})  e^{\hat H_3} = \hat D(-\mathbf t) e^{\hat H_3} \hat D(\mathbf t) e^{\frac{1}{4}(-\mathbf z +\mathbf{l})^T i\Omega W_3 (-\mathbf z +\mathbf{l})}$. The value of $\mathbf t$ is irrelevant.
     In the ninth line, we used the cyclic property of the trace which is allowed because $e^{\hat H_3}$ is trace class
 with $ W_3 = \frac{I+\exp(i\Omega H_3)}{I-\exp(i\Omega H_3)}$
and we have $W_3= - V_3 i\Omega$.

Let us now simplify the exponents. First, we can merge the product of the exponential functions into one single exponential function whose argument is:
\begin{align}
   E &:= -\frac{1}{4}\mathbf l^T i\Omega W_1 \mathbf l +(-\mathbf z)^T i\Omega \mathbf s_2  -\frac{1}{2} (-\mathbf z)^T i\Omega \mathbf l +\frac{1}{4}(-\mathbf z +\mathbf{l})^T i\Omega W_3 (-\mathbf z +\mathbf{l})  .
   \end{align}
Now we can simplify the last term:
\begin{align}
    \frac{1}{4}(-\mathbf z +\mathbf{l})^T i\Omega W_3 (-\mathbf z +\mathbf{l})   &= \frac{1}{4} \Big( +\mathbf z^T i\Omega W_3 \mathbf z +\mathbf l^T i\Omega W_3 \mathbf l - \mathbf z^T i\Omega W_3 \mathbf l - \mathbf l^T i\Omega W_3 \mathbf z \Big) \\
    &= \frac{1}{4} \Big( +\mathbf z^T i\Omega W_3 \mathbf z +\mathbf l^T i\Omega W_3 \mathbf l - \mathbf z^T i\Omega W_3 \mathbf l - (i\Omega W_3 \mathbf z)^T\mathbf l \Big) \\
    &= \frac{1}{4} \Big( +\mathbf z^T i\Omega W_3 \mathbf z +\mathbf l^T i\Omega W_3 \mathbf l - \mathbf z^T i\Omega W_3 \mathbf l -  \mathbf z^T (i\Omega W_3)^T \mathbf l \Big) .
\end{align}
With this we can rewrite
\begin{align}
    E= \underbrace{-\frac{1}{4}\mathbf l^T i\Omega W_1 \mathbf l + \frac{1}{4}  \mathbf l^T i\Omega W_3 \mathbf l}_{=: E_0} \underbrace{- \mathbf{z}^T \left( i\Omega \mathbf s_2 - \frac{1}{2} i\Omega \mathbf l+ \frac{1}{4}i\Omega W_3\mathbf l + \frac{1}{4}(i\Omega W_3)^T \mathbf l \right)}_{=: E_1} \underbrace{+ \frac{1}{4}\mathbf z^T i\Omega W_3 \mathbf z}_{=:E_2} .
\end{align}
Let us now simplify each term. 
First, we have
\begin{align}
    E_0 = -\frac{1}{4}\mathbf l^T i\Omega W_1 \mathbf l + \frac{1}{4}  \mathbf l^T i\Omega W_3 \mathbf l =  - \delta\mathbf s^T (V_1+V_2)^{-1} \delta\mathbf s ,
\end{align}
using Eq.~\eqref{eq:PropertyW}. 

Now, for the terms $E_1,E_2$ we note that
\begin{align}
    i\Omega W_3 = i\Omega (- V_3 i\Omega) = \Omega V_3\Omega = - \Omega^T V_3 \Omega .
\end{align}
We thus have
\begin{align}
    E_2 = + \frac{1}{4}\mathbf z^T i\Omega W_3 \mathbf z = - \frac{1}{4}\mathbf z^T \Omega^T V_3 \Omega \mathbf z .
\end{align}
Furthermore, we have
\begin{align}
    E_1 &= - \mathbf{z}^T \left( i\Omega \mathbf s_2 - \frac{1}{2} i\Omega \mathbf l+ \frac{1}{4}i\Omega W_3\mathbf l + \frac{1}{4}(i\Omega W_3)^T \mathbf l \right) \\
    &= - \mathbf{z}^T \left( -i\Omega^T \mathbf s_2 + \frac{1}{2} i\Omega^T \mathbf l+ \frac{1}{4} (- \Omega^T V_3 \Omega )\mathbf l + \frac{1}{4}(- \Omega^T V_3 \Omega )^T \mathbf l \right) \\
    &= - \mathbf{z}^T \left( -i\Omega^T \mathbf s_2 + \frac{1}{2} i\Omega^T \mathbf l - \frac{1}{2}   \Omega^T V_3 \Omega \mathbf l  \right) \\
     &=    \mathbf{z}^T \left( +i\Omega^T \mathbf s_2 - \frac{1}{2} i\Omega^T \mathbf l + \frac{1}{2}   \Omega^T V_3 \Omega \mathbf l  \right) \\
     &=    i\mathbf{z}^T \Omega^T\left( + \mathbf s_2 - \frac{1}{2}   \mathbf l + \frac{1}{2i}  V_3 \Omega \mathbf l  \right) \\
     &=    i\mathbf{z}^T \Omega^T\left( + \mathbf s_2 - \frac{1}{2}   \mathbf l - i\frac{1}{2}  V_3 \Omega \mathbf l  \right) .
\end{align}
In the third line, we used that $(\Omega^T V \Omega)^T = \Omega^T V \Omega$.

Now, recall $\mathbf l= (\exp[-i\Omega H_1]-I)\delta \mathbf s$. Let us rewrite $(\exp[-i\Omega H_1]-I)$. From Eq.~(31) of Ref.~\cite{seshadreesan2018renyi}, we know that  $ \exp[+i\Omega H] = (W-I) (W+I)^{-1}$. Thus we have $ \exp[-i\Omega H] = (W+I)(W-I)^{-1} $ and
\begin{align}
   \exp[-i\Omega H] - I &= (W+I)(W-I)^{-1} -I = (W+I)(W-I)^{-1} - (W-I)(W-I)^{-1} \\
   &= (W+I -(W-I))(W-I)^{-1}   = 2 (W-I)^{-1} = 2 (- Vi\Omega-I)^{-1} = -2 ( Vi\Omega +I)^{-1} \\
   &= -2 ( Vi\Omega + i\Omega i\Omega)^{-1} = -2 ( [V + i\Omega]i\Omega  )^{-1} \\
   &=  -2 (i\Omega)^{-1}( V + i\Omega  )^{-1} =  -2 \frac{1}{i}(-\Omega) ( V + i\Omega  )^{-1} = (-1)^2 2  (-i) \Omega ( V + i\Omega )^{-1} \\
   &= - 2  i \Omega ( V + i\Omega  )^{-1} ,
\end{align}
where we used $i\Omega i\Omega = i^2 \Omega^2 = +I$, and $\Omega^{-1}=-\Omega$.

    With this we find
\begin{align}
    - i\frac{1}{2}  V_3 \Omega \mathbf l &= - i\frac{1}{2}  V_3 \Omega  (\exp[-i\Omega H_1]-I)\delta \mathbf s \\
    &= - i\frac{1}{2}  V_3 \Omega  (-2 i\Omega (V_1+i\Omega)^{-1})\delta \mathbf s \\
    &= (-1)^2 \frac{2}{2} i^2 V_3 \Omega \Omega (V_1+i\Omega)^{-1}\delta \mathbf s \\
    &=     - V_3 \Omega \Omega (V_1+i\Omega)^{-1}\delta \mathbf s \\
    &=     + V_3  (V_1+i\Omega)^{-1}\delta \mathbf s \\
    &= + [ -i\Omega + (V_2+i\Omega) (V_2+V_1)^{-1} (V_1+i\Omega)]  (V_1+i\Omega)^{-1}\delta \mathbf s \\
    &= + [ -i\Omega (V_1+i\Omega)^{-1}  \delta \mathbf s + (V_2+i\Omega) (V_2+V_1)^{-1}   \delta \mathbf s]  \\
    &= + [+ \frac{1}{2}\mathbf l + (V_2+i\Omega) (V_2+V_1)^{-1}   \delta \mathbf s] \\
    &= + \frac{1}{2}\mathbf l + (V_2+i\Omega) (V_2+V_1)^{-1}   \delta \mathbf s ,
\end{align}
where we used $\Omega \Omega = -I$, and $V_3 = -i\Omega + (V_2+i\Omega) (V_2+V_1)^{-1} (V_1+i\Omega)$.

Finally, we obtain
\begin{align}
    E_1 &=  i\mathbf{z}^T \Omega^T\left( + \mathbf s_2 - \frac{1}{2}   \mathbf l - i\frac{1}{2}  V_3 \Omega \mathbf l  \right)  \\
    &= i\mathbf{z}^T \Omega^T\left( + \mathbf s_2 - \frac{1}{2}   \mathbf l + \frac{1}{2}\mathbf l + (V_2+i\Omega) (V_2+V_1)^{-1}   \delta \mathbf s  \right)\\
    &= i\mathbf{z}^T \Omega^T\left( + \mathbf s_2   + (V_2+i\Omega) (V_2+V_1)^{-1}   \delta \mathbf s  \right) .
\end{align}

We finally have simplified all the terms in $E$ sufficiently. We thus obtain
\begin{align}
    \chi_{\hat O} (\mathbf z) = e^{- \delta\mathbf s^T (V_1+V_2)^{-1} \delta\mathbf s} e^{- \frac{1}{4}\mathbf z^T \Omega^T V_3 \Omega \mathbf z } e^{i\mathbf{z}^T \Omega^T\left( + \mathbf s_2   + (V_2+i\Omega) (V_2+V_1)^{-1}   \delta \mathbf s  \right)} \mathrm{Tr}[e^{\hat H_3}] .
\end{align}
The value of $\mathrm{Tr}[e^{\hat H_3}]$ was calculated in Theorem~\ref{theorem:TraceValue}.
\end{proof}

\begin{corollary}\label{corollary:CharacteristicInverse}
     Consider the    operator $\hat O = \left( \hat D(-\mathbf s_1) e^{\hat H_1} \hat D(\mathbf s_1) \right)\left(  \hat D(-\mathbf s_2) e^{-\hat H_2} \hat D(\mathbf s_2) \right)$, which is the product of an unnormalized Gaussian state $ \hat D(-\mathbf s_1) e^{\hat H_1} \hat D(\mathbf s_1)$, with mean $\mathbf s_1$ and covariance matrix $V_1$, and the inverse $(\hat D(-\mathbf s_2) e^{\hat H_2} \hat D(\mathbf s_2))^{-1}=\hat D(-\mathbf s_2) e^{-\hat H_2} \hat D(\mathbf s_2)$ of an unnormalized Gaussian state $\hat D(-\mathbf s_2) e^{\hat H_2} \hat D(\mathbf s_2)$ with mean $\mathbf s_2$ and covariance matrix $V_2$. Let 
     \begin{align}
        V_{2} - V_{1} +\Omega^T (V_{2}^{-1} - V_{1}^{-1}) \Omega > 0 .
    \end{align}
     The characteristic function of $\hat O$ is given by
   \begin{align}
    \chi_{\hat O} (\mathbf z) = e^{- \delta\mathbf s^T (V_1-V_2)^{-1} \delta\mathbf s} e^{- \frac{1}{4}\mathbf z^T \Omega^T V_3 \Omega \mathbf z } e^{i\mathbf{z}^T \Omega^T\left( + \mathbf s_2   + (-V_2+i\Omega) (-V_2+V_1)^{-1}   \delta \mathbf s  \right)} \sqrt{\frac{\mathrm{Det}[\frac{V_1+i\Omega}{2}]\mathrm{Det}[\frac{V_2+i\Omega}{2}]}{ \mathrm{Det}[\frac{V_2-V_1}{2}]}} ,
\end{align}
where
\begin{align}
    V_3 = -i\Omega  + (-V_2+i\Omega) (V_1-V_2)^{-1}(V_1+i\Omega) .
\end{align}
\end{corollary}
\begin{proof}
    The proof is completely analogous to the proof of Theorem~\ref{theorem:CharacteristicProd}. The condition $ V_{2} - V_{1} +\Omega^T (V_{2}^{-1} - V_{1}^{-1})\Omega>0$ ensures that $e^{\hat H_1}e^{-\hat H_2}$ is trace class, which allows us to use the cyclic property of the trace as in the  proof of Theorem~\ref{theorem:CharacteristicProd}. The condition $ V_{2} - V_{1} +\Omega^T (V_{2}^{-1} - V_{1}^{-1})\Omega>0$ also ensures  that the inverse  $(V_{2} - V_{1})^{-1}$ exists which follows from Corollary~\ref{corollary:Invertibility}.
    As shown in Ref.~\cite{seshadreesan2018renyi}, all identities used in the proof of Theorem~\ref{theorem:CharacteristicProd} are still valid under $e^{\hat H_2} \rightarrow e^{-\hat H_2}$ if we make the exchanges $V_2 \rightarrow -V_2$ and $W_2\rightarrow -W_2$. Note that $\mathrm{Det}[\frac{-V_2+i\Omega}{2}] = \mathrm{Det}[\frac{V_2+i\Omega}{2}]$.
\end{proof}

\section{Classical Barankin bound of Gaussian probability density}\label{app:ClassicalBarankinGaussian}
Let us here derive classical Barankin bound of classical Gaussian probability density. This allows us to compare the classical result to its quantum generalization derived in Section~\ref{sec:Gaussian}.
First, we define the family $\{p(\mathbf x\vert \tilde\lambda)\}_{\tilde\lambda\in\Lambda}$ of multivariate Gaussian distributions
\begin{align}
    p(\mathbf x\vert \tilde\lambda) = \frac{\exp\left( - \frac{1}{2} (\mathbf  x- \mathbf s_{\tilde\lambda})^T V_{\tilde\lambda}^{-1} (\mathbf  x- \mathbf s_{\tilde\lambda})\right)}{\sqrt{(2\pi)^{m}\mathrm{Det}[V_{\tilde\lambda}]}} ,
\end{align}
where $\mathbf x \in \mathds R^{m}$, $\mathbf s_{\tilde\lambda} \in \mathds R^{m}$ and $V_{\tilde\lambda}$ is a $m\times m$ real symmetric positive definite matrix. Let us now calculate the classical Barankin matrix
\begin{align}
     (C_{\mathrm{BB}})_{kl} &= \int_{\mathds R^{m}}\mathrm{d}^{m}x\,\frac{p(\mathbf x\vert\lambda_k)p(\mathbf x\vert\lambda_l)}{p(\mathbf x\vert\lambda)} \\
     &= \sqrt{\frac{(2\pi)^{m}\mathrm{Det}[V_{\lambda}]}{ (2\pi)^{m}\mathrm{Det}[V_{\lambda_k}] (2\pi)^{m}\mathrm{Det}[V_{\lambda_l}]}} \int_{\mathds R^{m}}\mathrm{d}^{m}x\,\frac{\exp\left( - \frac{1}{2} (\mathbf  x- \mathbf s_{ \lambda_k})^T V_{ \lambda_k}^{-1} (\mathbf  x- \mathbf s_{ \lambda_k})\right) \exp\left( - \frac{1}{2} (\mathbf  x- \mathbf s_{ \lambda_l})^T V_{ \lambda_l}^{-1} (\mathbf  x- \mathbf s_{ \lambda_l})\right)}{\exp\left( - \frac{1}{2} (\mathbf  x- \mathbf s_{ \lambda})^T V_{ \lambda}^{-1} (\mathbf  x- \mathbf s_{ \lambda})\right)} \\
     &= \sqrt{\frac{\mathrm{Det}[V_{\lambda}]}{ (2\pi)^{m}\mathrm{Det}[V_{\lambda_k}]  \mathrm{Det}[V_{\lambda_l}]}} \int_{\mathds R^{m}}\mathrm{d}^{m}x\, \exp\left(-\mathbf x^T A \mathbf x+ \mathbf x^T \mathbf b + c\right) ,
\end{align}
where $A= \frac{1}{2} (V_{\lambda_k}^{-1}+V_{\lambda_l}^{-1} - V_\lambda^{-1})$, $\mathbf b = V_{\lambda_k}^{-1} \mathbf{s_{\lambda_k}} + V_{\lambda_l}^{-1} \mathbf{s_{\lambda_l}} -  V_{\lambda}^{-1} \mathbf{s_{\lambda }}$ and $c= -\frac{1}{2}(\mathbf s_{\lambda_k}^T V_{\lambda_k}^{-1}\mathbf s_{\lambda_k}+\mathbf s_{\lambda_l}^T V_{\lambda_l}^{-1}\mathbf s_{\lambda_l}-\mathbf s_{\lambda}^T V_{\lambda}^{-1}\mathbf s_{\lambda})$. Next, we use $\int_{\mathds R^{m}}\mathrm d^{m} x\, e^{-\mathbf x^T A \mathbf x+ \mathbf x^T \mathbf b} = \frac{\pi^{m/2}}{\sqrt{\mathrm{Det}[A]}} e^{\frac{1}{4}\mathbf b^T A^{-1}\mathbf b}$ which is valid when $A>0$, or equivalently $V_{\lambda_k}^{-1}+V_{\lambda_l}^{-1} - V_\lambda^{-1} >0$, from which follows
\begin{align}
     (C_{\mathrm{BB}})_{kl} &= \sqrt{\frac{\mathrm{Det}[V_{\lambda}]}{ (2\pi)^{m}\mathrm{Det}[V_{\lambda_k}]  \mathrm{Det}[V_{\lambda_l}]}} \frac{\pi^{m/2}}{\sqrt{\mathrm{Det}[A]}} e^{c} e^{\frac{1}{4}\mathbf b^T A^{-1}\mathbf b}  \\
     &= \sqrt{\frac{\mathrm{Det}[V_{\lambda}]}{ (2\pi)^{m}\mathrm{Det}[V_{\lambda_k}]  \mathrm{Det}[V_{\lambda_l}]}} \frac{\pi^{m/2}}{\sqrt{\mathrm{Det}[\frac{1}{2} (V_{\lambda_k}^{-1}+V_{\lambda_l}^{-1} - V_\lambda^{-1})]}}  e^{-\frac{1}{2}(\mathbf s_{\lambda_k}^T V_{\lambda_k}^{-1}\mathbf s_{\lambda_k}+\mathbf s_{\lambda_l}^T V_{\lambda_l}^{-1}\mathbf s_{\lambda_l}-\mathbf s_{\lambda}^T V_{\lambda}^{-1}\mathbf s_{\lambda})} \nonumber \\
     & \quad\quad \quad \quad \quad\quad \quad \quad \quad\quad \quad \quad \quad\times  e^{\frac{1}{4} (V_{\lambda_k}^{-1} \mathbf{s_{\lambda_k}} + V_{\lambda_l}^{-1} \mathbf{s_{\lambda_l}} -  V_{\lambda}^{-1} \mathbf{s_{\lambda }})^T (\frac{1}{2} (V_{\lambda_k}^{-1}+V_{\lambda_l}^{-1} - V_\lambda^{-1}))^{-1}(V_{\lambda_k}^{-1} \mathbf{s_{\lambda_k}} + V_{\lambda_l}^{-1} \mathbf{s_{\lambda_l}} -  V_{\lambda}^{-1} \mathbf{s_{\lambda }})} \\
     &= \sqrt{\frac{\mathrm{Det}[V_{\lambda}]}{  \mathrm{Det}[V_{\lambda_k}]  \mathrm{Det}[V_{\lambda_l}]\mathrm{Det}[V_{\lambda_k}^{-1}+V_{\lambda_l}^{-1} - V_\lambda^{-1}]}}     e^{-\frac{1}{2}(\mathbf s_{\lambda_k}^T V_{\lambda_k}^{-1}\mathbf s_{\lambda_k}+\mathbf s_{\lambda_l}^T V_{\lambda_l}^{-1}\mathbf s_{\lambda_l}-\mathbf s_{\lambda}^T V_{\lambda}^{-1}\mathbf s_{\lambda})} \nonumber \\
     & \quad\quad \quad \quad \quad\quad \quad \quad \quad\quad \quad \quad \quad\times  e^{\frac{1}{2} (V_{\lambda_k}^{-1} \mathbf{s_{\lambda_k}} + V_{\lambda_l}^{-1} \mathbf{s_{\lambda_l}} -  V_{\lambda}^{-1} \mathbf{s_{\lambda }})^T (  V_{\lambda_k}^{-1}+V_{\lambda_l}^{-1} - V_\lambda^{-1})^{-1}(V_{\lambda_k}^{-1} \mathbf{s_{\lambda_k}} + V_{\lambda_l}^{-1} \mathbf{s_{\lambda_l}} -  V_{\lambda}^{-1} \mathbf{s_{\lambda }})} .
\end{align}
By comparing this classical Barankin matrix to the quantum version given in Eq.~\eqref{eq:QBBgaussianFormaula}, we clearly notice a structural similarity.

\bibliography{barankin}

\end{document}